\documentclass[twocolumn]{IEEEtranTCOM}
\usepackage{amsmath,subfigure}
\usepackage{graphicx}
\usepackage{grffile}
\usepackage{algpseudocode}
\usepackage{algorithm}
\usepackage{epstopdf}
\usepackage{amsmath}
\usepackage{amssymb}
\usepackage{breqn}

\newtheorem{theorem}{Theorem}[section]
\newtheorem{lemma}[theorem]{Lemma}

\newtheorem{mydef}{Definition}

\newenvironment{proof}[1][Proof]{\begin{trivlist}
\item[\hskip \labelsep {\bfseries #1}]}{\end{trivlist}}

\newcommand{\qed}{\nobreak \ifvmode \relax \else
      \ifdim\lastskip<1.5em \hskip-\lastskip
      \hskip1.5em plus0em minus0.5em \fi \nobreak
      \vrule height0.75em width0.5em depth0.25em\fi}

\begin{document}
\title{Message Passing Algorithms for Phase Noise Tracking Using
Tikhonov Mixtures}

\author{Shachar~Shayovitz,~\IEEEmembership{Student~Member,~IEEE,}
        and~Dan~Raphaeli,~\IEEEmembership{Member,~IEEE}
\thanks{S. Shayovitz and D. Raphaeli are with the Department
of EE-Systems, Tel Aviv University, Tel Aviv,
Israel, e-mail: shachars@post.tau.ac.il,danr@eng.tau.ac.il.}}  %

\markboth{IEEE Transactions on Communications}%
{Submitted paper}
\maketitle

\begin{abstract}
In this work, a new low complexity iterative algorithm for decoding data transmitted over strong phase
noise channels is presented. The algorithm is based on the Sum \& Product Algorithm (SPA) with phase noise
messages modeled as Tikhonov mixtures. Since mixture based Bayesian inference such as SPA, creates an
exponential increase in mixture order for consecutive messages, mixture reduction is necessary. We propose
a low complexity mixture reduction algorithm which finds a reduced order mixture whose dissimilarity
metric is mathematically proven to be upper bounded by a given threshold.
As part of the mixture reduction, a new method for optimal clustering provides the closest circular
distribution, in Kullback Leibler sense, to any circular mixture. We further show a method for limiting
the number of tracked components and further complexity reduction approaches. We show simulation results
and complexity analysis for the proposed algorithm and show better performance than other state of the art
low complexity algorithms.
We show that the Tikhonov mixture approximation of SPA messages is equivalent to the tracking of multiple
phase trajectories, or also can be looked as smart multiple phase locked loops (PLL). When the number of
components is limited to one the result is similar to a smart PLL.
\end{abstract}
\begin{keywords}
phase noise, factor graph, Tikhonov, cycle slip, directional
statistics, moment matching,mixture models
\end{keywords}

\IEEEpeerreviewmaketitle

\section{Introduction}
\label{sec:intro}
\IEEEPARstart{M}any high frequency communication systems operating
today employ low cost upconverters or downconverters which create
phase noise. Phase noise can severely limit the information rate of a communications
system and pose a serious challenge for the detection systems.
Moreover, simple solutions for phase noise tracking such as PLL either
require low phase noise or otherwise require many pilot symbols which
reduce the effective data rate.

In the last decade we have witnessed a significant amount of research
done on joint estimation and decoding of phase noise and coded
information. For example, \cite{barb2005} and \cite{colavolpe2006} which are based on the
factor graph representation of the joint posterior, proposed in \cite{worthen2001} and allows the design
of efficient message passing algorithms which incorporate both the code
graph and the channel graph. The use of LDPC or Turbo decoders, as
part of iterative message passing schemes, allows the receiver to
operate in low SNR regions while requiring less pilot symbols.

In order to perform MAP decoding of the code symbols, the SPA is applied to the factor graph. The SP
algorithm is a message passing algorithm which computes the exact
marginal for each code symbol, provided there are no cycles in the
factor graph. In the case of phase noise channels, the messages related to the phase are continuous, thus
recursive computation of messages requires computation of integrals which have no analytical solution and
the direct application of this algorithm is not feasible.
A possible approximation of MAP detection is to quantize the phase noise and perform an
approximated SP. The channel phase takes only a finite number of
values $L$, thus creating a trellis diagram representing the random
walk. If we suppose a forward - backward scheduling, the SPA reduces to a BCJR run on this trellis
following LDPC decoding. This algorithm (called DP - discrete phase in this
paper) requires large computational resources (large $L$) to reach
high accuracy, rendering it not practical for some real world
applications.

In order to circumvent the problem of continuous messages, many algorithms have resorted to
approximations. In \cite{colavolpe2006}, the algorithm uses channel memory truncation
rather than an explicit representation of the channel parameters. In
\cite{barb2005} section B., an algorithm which efficiently balances
the tradeoff between accuracy and complexity was proposed (called BARB
in this paper). BARB uses Tikhonov distribution parameterizations
(canonical model) for all the SPA messages concerning a phase node.
However, the approximation as defined in \cite{barb2005}, is only good
when the information from the LDPC decoder is good (high reliability).
In the first iteration the approximation is poor, and in fact
exists only for pilot symbols. The LLR messages related to the
received symbols which are not pilots are essentially zero (no
information). This inability to accurately approximate the messages in
the first iterations causes many errors and can create an error floor.
This problem is intensified when using either low code rate or high
code rate. In the first case, it is since the pilots are less
significant, since their energy is reduced. In the second case, the
poor estimation of the symbols far away from the pilots cannot be
overcome by the error correcting capacity of the code.
In order to overcome this limitation, BARB relies on the insertion of
frequent pilots to the transmitted block causing a reduction of the
information rate.

In this paper, a new approach for approximating the phase noise forward and backward messages using
Tikhonov mixtures is proposed. Since SP recursion equations create an exponential increase in the number
of mixture components, a mixture reduction algorithm is needed at each phase message calculation to keep
the mixture order small.
We have tested few state of the art clustering algorithms, and those algorithms failed for this task, and cannot provide proven accuracy. Therefore we have derived a new clustering algorithm. A distinct property of
the new algorithm is its ability to provide adaptive mixture order, while keeping specified accuracy
constraint, where the accuracy is the Kullback Leibler (KL) divergence between the original and the
clustered pdfs. A proof for the accuracy of this mixture reduction algorithm is also presented in this
paper.
We show that the process of hypothesis expansion followed by clustering is equivalent to a sophisticated
tracker which can track most of the multiple hypotheses of possible phase trajectories. Occasionally, the
number of hypotheses grows, and more options for phase trajectories emerge. Each such event causes the
tracker to create another tracking loop. In other occasions, two trajectories are merged into one. We
show, as an approximation, the tracking of each isolated phase trajectory is equivalent to a PLL and a
split event is equivalent to a point in time when a phase slip may happen.

In the second part, we use a limited order Tikhonov mixture. This limitation may cause the tracking
algorithm to lose tracking of the correct phase trajectory, and is analogous to a cycle slip in PLL. We
propose a method to combat these slips with only a slight increase in complexity. The principle operation
of the method is that each time some hypothesis is abandoned, we can calculate the probability of being in
the correct trajectory and we can use this information wisely in the calculation of the messages. We
provide further complexity reduction approaches. One of these approaches is to abandon the clustering
altogether, and replace it by component selection algorithm, which maintains the specified accuracy but
requires more components in return. Now the complexity of clustering is traded against the complexity of
other tasks. Finally, we show simulations results which demonstrate that the proposed scheme's Packet
Error Rate (PER) are comparable to the DP algorithm and that the resulting computational complexity is
much lower than DP and in fact is comparable to the algorithm proposed in \cite{barb2005}.

The reminder of this paper is organized as follows. Section II
introduces the channel model and presents the derivation of the exact
SPA  from \cite{barb2005}. In Section III, we introduce the reader to
the directional statistics framework, and some helpful results on the
KL divergence. Section IV presents the mixture order canonical
model and provides a review on mixture reduction algorithms. Section V presents two mixture reduction
algorithms for approximating the SP messages. Section VI presents the computation of LLRs. A complexity
comparison is carried out in Section VII. Finally, in
Section VIII we present some numerical results and in Section IX, we
discuss the results and point out some interesting claims.

\section{System Model}
\label{sec:system_model}
In this section we present the system model used throughout this paper. We assume a sequence of data bits
is encoded using an LDPC code and then mapped to a complex signal constellation $\mathbb{A}$ of size $M$,
resulting in a sequence of complex modulation
symbols $\mathbf{c} = (c_{0},c_{1},...,c_{K-1})$. This sequence is transmitted
over an AWGN channel affected by carrier phase noise. Since we use a long LDPC code, we can assume the
symbols are drawn independency from the constellation. The
discrete-time baseband complex equivalent channel model at the
receiver is given by:
\begin{equation}\label{sys_model}
    r_{k} = c_{k}e^{j\theta_{k}}+n_{k} \;\;\;\;  k=0,1,...,K-1.
\end{equation}
where $K$ is the length of the transmitted sequence of complex symbols.
The phase noise stochastic model is a Wiener process
\begin{equation}\label{weiner}
    \theta_{k} = \theta_{k-1} + \Delta_{k}
\end{equation}
where ${\Delta_{k}}$ is a real, i.i.d gaussian sequence with
$\Delta_{k} \sim \mathcal{N}(0,\sigma_{\Delta}^{2})$ and $\theta_{0}
\sim \mathcal{U}[0,2\pi)$. For the sake of clarity we define pilots as
transmitted symbols which are known to both the transmitter and
receiver and are repeated in the transmitted block every known number
of data symbols. We also define a preamble as a sequence of pilots in
the beginning of a transmitted block. We assume that the transmitted sequence is padded with pilot symbols
in order to bootstrap the algorithms and maintain the tracking.

\subsection{Factor Graphs and the Sum Product Algorithm}
Since we are interested in optimal MAP detection, we will use the framework defined in \cite{worthen2001},
compute the SPA equations and thus perform approximate MAP detection. The factor graph representation of
the joint posterior distribution
was given in \cite{barb2005} and is shown in Fig. \ref{fig:fg}.
\begin{figure}
  \centering
  \includegraphics[width=8cm]{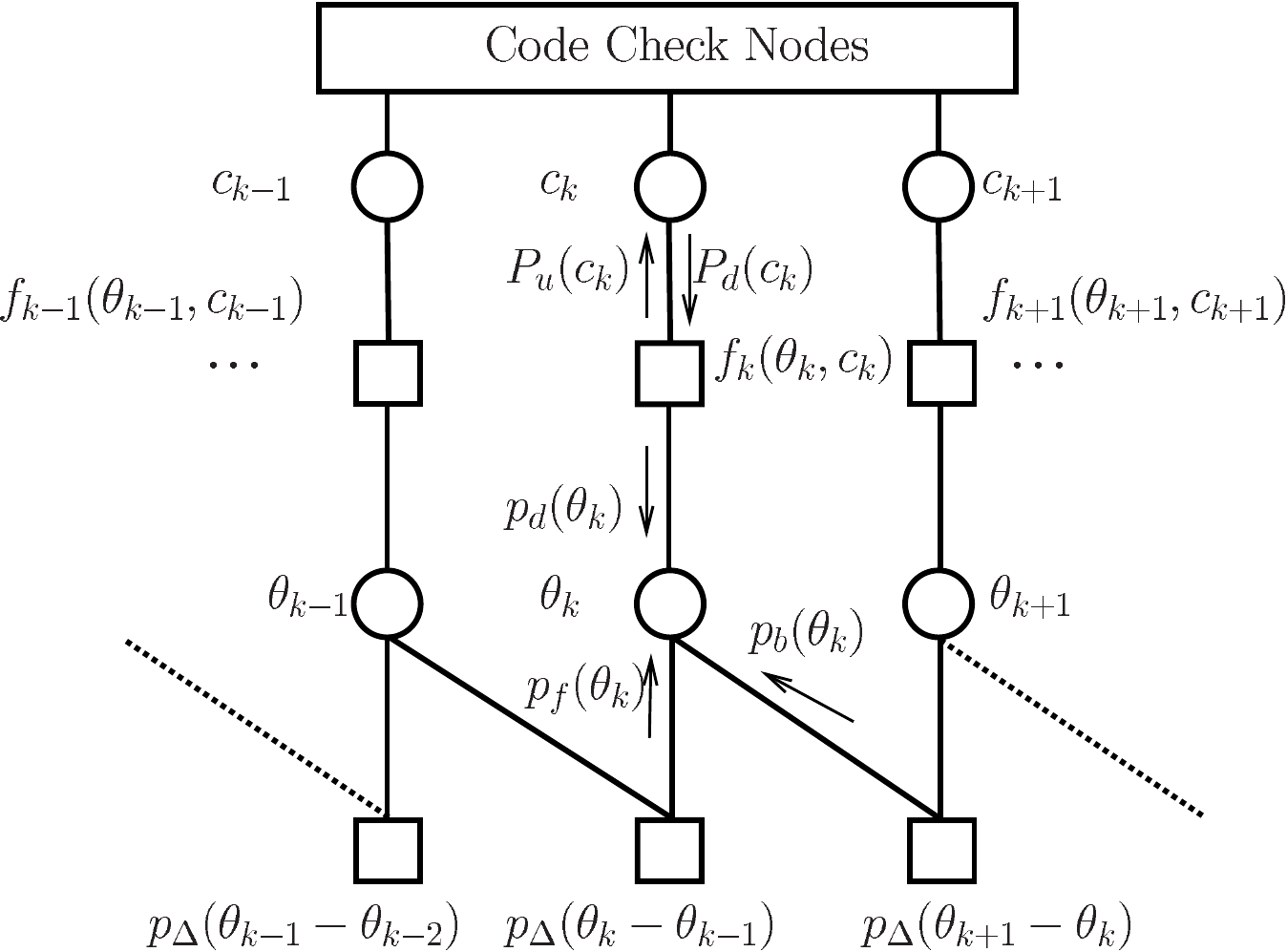}\\
  \caption{Factor graph representation of the joint posterior
distribution}\label{fig:fg}
\end{figure}
The resulting Sum \& Product messages are computed by
\begin{equation}\label{pf}
    p_{f}(\theta_{k}) =
\int_{0}^{2\pi}p_{f}(\theta_{k-1})p_{d}(\theta_{k-1})p_{\Delta}(\theta_{k}-\theta_{k-1})d\theta_{k-1}
\end{equation}
\begin{equation}\label{pb}
    p_{b}(\theta_{k}) =
\int_{0}^{2\pi}p_{b}(\theta_{k+1})p_{d}(\theta_{k+1})p_{\Delta}(\theta_{k+1}-\theta_{k})d\theta_{k+1}
\end{equation}
\begin{equation}\label{pd}
    p_{d}(\theta_{k}) = \sum_{x \in \mathbb{A}} P_{d}(c_{k}=x) e_{k}(c_{k},\theta_{k})
\end{equation}
\begin{equation}\label{Pu}
    P_{u}(c_{k}) =
\int_{0}^{2\pi}p_{f}(\theta_{k})p_{b}(\theta_{k})e_{k}(c_{k},\theta_{k})d\theta_{k}
\end{equation}
\begin{equation}\label{fk}
    e_{k}(c_{k},\theta_{k}) \propto
\exp\{-\frac{|r_{k}-c_{k}e^{j\theta_{k}}|^{2}}{2\sigma^{2}}\}
\end{equation}
\begin{equation}\label{p_del}
    p_{\Delta}(\theta_{k}) =
\sum^{\infty}_{l=-\infty}g(0,\sigma_{\Delta}^{2},\theta_{k}-l2\pi)
\end{equation}
Where $r_{k}$,$P_{d}$, $\sigma^{2}$ and
$g(0,\sigma_{\Delta}^{2},\theta)$ are the
received base band signal, symbol soft information from LDPC decoder,
AWGN variance and Gaussian distribution, respectively. The messages
$p_{f}(\theta_{k})$ and $p_{b}(\theta_{k})$ are called in this paper
the forward and backward phase noise SP messages, respectively.

The detection process starts with the channel section providing the first LLRs ($P_{u}(c_{k})$) to the
LDPC decoder, and so on. A different scheduling could be applied on a general setting, but this will not
be possible with the algorithms in this paper. Due to the fact that the phase symbols are continuous
random variables, a direct implementation of these equations is not possible
and approximations are unavoidable. Assuming enough quantization
levels, the DP algorithm can approximate the above equations as close
as we wish. However, this algorithm requires large computational
resources to reach high accuracy, rendering it not practical for some
real world applications. In \cite{shachar2012},\cite{shachar_multi2012} and \cite{shachar_old2012}, modified Tikhonov
approximations were used for the messages in the SPA which lead to a
very simple and fast algorithm. In this paper, an approximate
inference algorithm is proposed which better balances the tradeoff
between accuracy and complexity for strong phase noise channels.
\section{Preliminaries}

\subsection{Directional Statistics}
\label{sec:pagestyle}
Directional statistics is a branch of mathematics which studies random
variables defined on circles and spheres. For example, the probability
of the wind to blow at a certain direction. The \emph{circular} mean
and variance of a circular random variable $\theta$, are defined in
\cite{mardia2000}, as
\begin{equation}\label{circ_mu}
    \mu_{C} = \angle \mathbb{E} (e^{j\theta})
\end{equation}
\begin{equation}\label{circ_var}
    \sigma^{2}_{C} = \mathbb{E}(1-cos(\theta-\mu_{C}))
\end{equation}

One can see that for small angle variations around the \emph{circular}
mean, the definition of the \emph{circular} variance coincides with
the standard definition of the variance of a random variable defined
on the real axis, since $1-cos(\theta-\mu_{C}) \approx
(\theta-\mu_{C})^2$.
One of the most commonly used circular distributions is the Tikhonov
distribution and is defined as,
\begin{equation}\label{tikh_define}
   g(\theta) = \frac{e^{Re[\kappa_{g}e^{-j(\theta-\mu_{g})}]}}{2\pi
I_{0}(\kappa_{g})}
\end{equation}
According to (\ref{circ_mu}) and (\ref{circ_var}), the \emph{circular}
mean and \emph{circular} variance of a Tikhonov distribution are,
\begin{equation}\label{tikh_mu}
    \mu_{C} = \mu_{g}
\end{equation}
\begin{equation}\label{tikh_var}
    \sigma^{2}_{C} = 1-\frac{I_{1}(\kappa_{g})}{I_{0}(\kappa_{g})}
\end{equation}
where $I_{0}(x)$ and $I_{1}(x)$ are the  modified Bessel function of
the first kind of the zero and first order, respectively.
An alternative formulation for the Tikhonov pdf uses a single complex
parameter $z = \kappa_{g} e^{j\mu_{g}}$
residual phase noise in a first order PLL when the input phase noise is constant is the tikhonov distribtion

\subsection{Circular Mean \& Variance Matching}
\label{ssec:subhead}
In this section we will present a new theorem in directional statistics. The theorem states that the
nearest Tikhonov distribution, $g(\theta)$, to any circular
distribution,$f(\theta)$ (in a Kullback Liebler (KL) sense), has its
circular mean and variance matched to those of the circular
distribution .
The Kullback Liebler (KL) divergence is a common information theoretic
measure of similarity between probability distributions, and is
defined as \cite{KL1951},

\begin{equation}\label{KL1}
    D(f||g) \triangleq \int_0^{2\pi}f(\theta)\log
\frac{f(\theta)}{g(\theta)} d\theta
\end{equation}

\begin{mydef}
We define the operator $g(\theta) = \textsf{CMVM}[f(\theta)]$
(Circular Mean and Variance Matching), to take a circular pdf -
$f(\theta)$ and create a Tikhonov pdf $g(\theta)$ with the same
circular mean and variance.
\end{mydef}

\begin{theorem}
\emph{(CMVM):}
\label{mix_tikh_thr}
Let $f(\theta)$ be a circular distribution, then the Tikhonov
distribution $g(\theta)$ which minimizes $D(f||g)$ is,
\begin{equation}\label{CMVM_thr}
     g(\theta) = \textsf{CMVM}[f(\theta)]
\end{equation}
\end{theorem}
The proof can be found in appendix \ref{sec:cmvm_thr}.

\subsection{Helpful Results for KL Divergence}
We introduce the reader to three results related to the Kullback-Leibler Divergence which
will prove helpful in the next sections.

\begin{lemma}\label{kl_bound1}
Suppose we have two distributions, $f(\theta)$ and $g(\theta)$,
\[f(\theta) = \sum_{i=1}^{M}\alpha_{i}f_{i}(\theta)\]
\begin{equation}
     D_{KL}(\sum_{i=1}^{M}\alpha_{i}f_{i}(\theta) || g(\theta)) \leq
\sum_{i=1}^{M}\alpha_{i}D_{KL}(f_{i}(\theta) || g(\theta))
\end{equation}
\end{lemma}
The proof of this bound can be found in \cite{runnalls2007} and is
based on the Jensen inequality.

\begin{lemma}\label{kl_bound2}
Suppose we have three distributions, $f(\theta)$ ,$g(\theta)$ and
$h(\theta)$. We define the following mixtures,

\begin{equation}\label{11}
   f_{1}(\theta) = \alpha f(\theta) + (1-\alpha)g(\theta)
\end{equation}

\begin{equation}\label{12}
   f_{2}(\theta) = \alpha f(\theta) + (1-\alpha)h(\theta))
\end{equation}
for $0 \leq \alpha \leq 1$

Then,

\begin{equation}
     D_{KL}(f_{1}(\theta) || f_{2}(\theta)) \leq
(1-\alpha)D_{KL}(g(\theta) || h(\theta))
\end{equation}
\end{lemma}

The proof for this identity can also be found in \cite{runnalls2007}.

\begin{lemma}\label{kl_bound3}
Suppose we have two mixtures, $f(\theta)$ and $g(\theta)$, of the same order $M$,
\[f(\theta) = \sum_{i=1}^{M}\alpha_{i}f_{i}(\theta)\]
and
\[g(\theta) = \sum_{j=1}^{M}\beta_{i}g_{i}(\theta)\]

Then the KL divergence between them can be upper bounded by,
\begin{equation}
     D_{KL}(f(\theta) || g(\theta)) \leq
D_{KL}(\alpha || \beta) + \sum_{i=1}^{M}\alpha_{i}D_{KL}(f_{i}(\theta) || g_{i}(\theta))
\end{equation}
\end{lemma}
where $D_{KL}(\alpha || \beta)$ is the KL divergence between the probability mass functions defined by all
the coefficients $\alpha_{i}$ and $\beta{i}$. The proof of this bound uses the sum log inequality and can
be found in \cite{Minh2003}.

\section{Tikhonov Mixture Canonical Model}
In this section we will present the Tikhonov mixture canonical model
for approximating the forward and backward phase noise SP messages. Firstly, we will give insight to the
motivation of using a mixture model for $p_{f}(\theta_{k})$ and $p_{b}(\theta_{k})$. The message, $p_{f}(\theta_{k})$, is the posterior phase distribution given the causal information $(r_{0},...,r_{k-1})$. If
we look at the (local) maximum over time we observe a phase trajectory. A phase trajectory is an
hypothesis about the phase noise process given the data. In case of zero a priori information, there will
be a $\frac{2\pi}{M}$ ambiguity in the phase trajectory, i.e. there will be $M$ parallel phase
trajectories with $\frac{2\pi}{M}$ separation between them.

Having a priori information on the data, such as preamble or pilots, can strengthen the correct hypothesis
and gradually remove wrong trajectories. However, as we get far away from the known data, more hypotheses
emerge. This dynamics is illustrated in Fig. \ref{fig:splits} where we have plotted in three dimensions
the forward phase noise messages ($p_{f}(\theta_{k})$) of the DP algorithm. The DP algorithm computes the
phase forward messages (\ref{pf}) on a quantized phase space. The axes
represent the time sample index, the quantized phase
for each symbol and the Z-axis is the posterior probability. In this figure there is only a small preamble
in the beginning and the end of the block and thus the first forward
messages are single mode Tikhonov distributions, which form a single
trajectory in the beginning of the figure and converges to a single trajectory in the end.
After the preamble, due to additive noise and phase noise, occasionally the algorithm cannot decide which
is the correct phase trajectory due to ambiguity in the symbols, thus it suggests to continue with two
trajectories each with its relative probability of occurring. This point is a split in the phase
trajectories and is analogous to a cycle slip in a PLL. If we approximate the messages at each point in
time as a a Tikhonov mixture with varying order, then each time we have a split, more components are added
to the mixture, and each time there is a merge, the number of components decreases.
This understating of the underlying structure of the phase messages is one of the most important
contributions of this paper and is the basis of the mixture model approach.

\begin{figure}
  \centering
  \includegraphics[width=8.5cm]{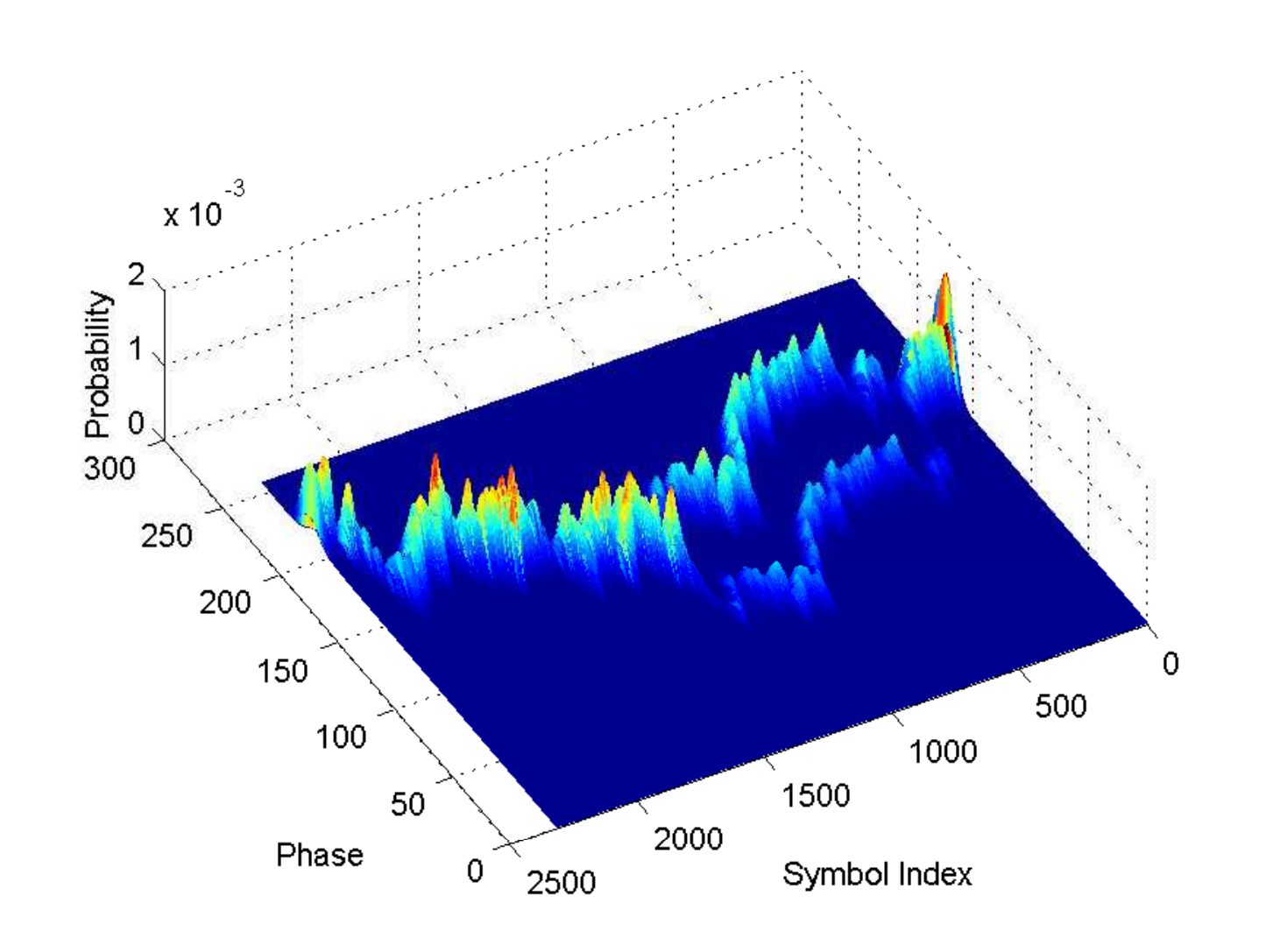}\\
  \caption{SP Phase Noise Forward Messages}\label{fig:splits}
\end{figure}

The advantage of using mixtures is in the ability to track several
phase trajectories simultaneously and provide better extrinsic
information to the LDPC decoder, which in turn will provide better
information on the code symbols to the phase estimator. In this way
the joint detection and estimation will converge quickly and avoid
error floors. However, as will be shown in a later section, the approximation of SP
messages using mixtures is a very difficult task since the mixture
order increases exponentially as we progress the phase tracking along
the received block. Therefore, there is a need for an efficient dimension reduction algorithm. In the following sections we
will propose a mixture reduction algorithm for the adaptive mixture
model. But first we will formulate the mixture reduction task
mathematically and describe algorithms which attempt to accomplish this task.

\subsection{Mixture Reduction - Problem Formulation}
As proposed above, the forward and backward messages are approximated
using Tikhonov mixtures,

\begin{equation}\label{new_pf}
    p_{f}(\theta_{k}) = \sum_{i
=1}^{N_{f}^{k}}\alpha^{k,f}_{i}t^{k,f}_{i}(\theta_{k})
\end{equation}
\begin{equation}\label{new_pb}
    p_{b}(\theta_{k}) = \sum_{i
=1}^{N_{b}^{k}}\alpha^{k,b}_{i}t^{k,b}_{i}(\theta_{k})
\end{equation}

where:
\begin{equation}\label{f_f_i}
t^{k,f}_{i}(\theta_{k}) =  \frac{e^{Re[z^{k,f}_{i}
e^{-j\theta_{k}}]}}{2\pi I_{0}(|z^{k,f}_{i}|)}
\end{equation}
\begin{equation}\label{f_b_i}
t^{k,b}_{i}(\theta_{k}) =  \frac{e^{Re[z^{k,b}_{i}
e^{-j\theta_{k}}]}}{2\pi I_{0}(|z^{k,b}_{i}|)}
\end{equation}

and,
$\alpha^{k,f}_{i}$ ,$\alpha^{k,b}_{i}$,$z^{k,f}_{i}$,$z^{k,b}_{i}$ are the
mixture coefficients and Tikhonov parameters of the forward and
backward messages of the phase sample.
If we insert approximations (\ref{new_pf}) and (\ref{new_pb}) in to
the forward and backward recursion equations (\ref{pf}) and (\ref{pb})
respectively, we get,

\begin{multline}\label{pf_eq}
    \tilde{p}_{f}(\theta_{k}) =
    \sum_{i =1}^{N_{f}^{k-1}}\int_{0}^{2\pi}\alpha^{k-1,f}_{i}t^{k-1,f}_{i}(\theta_{k-1})p_{d}(\theta_{k-1})\\ p_{\Delta}(\theta_{k}-\theta_{k-1})d\theta_{k-1}
\end{multline}
\begin{multline}\label{pb_eq}
    \tilde{p}_{b}(\theta_{k}) = \sum_{i=1}^{N_{b}^{k+1}}\int_{0}^{2\pi}\alpha^{k+1,b}_{i}t^{k+1,b}_{i}(\theta_{k+1})p_{d}(\theta_{k+1})\\ p_{\Delta}(\theta_{k+1}-\theta_{k})d\theta_{k+1}
\end{multline}

It is shown in \cite{barb2005} that the convolution of a Tikhonov and a Gaussian distributions is a

Tikhonov distribution,
\begin{equation}\label{pf_eq1}
    \tilde{p}_{f}(\theta_{k}) =
    \sum_{i =1}^{N_{f}^{k-1}}\sum_{x \in \mathbb{A}}\alpha^{k-1,f}_{i}\lambda^{k-1,f}_{i,x}\frac{e^{Re[\gamma(\sigma_{\Delta},\tilde{Z}^{k-1,f}_{i,x})e^{-j\theta_{k}}]}}{2\pi I_{0}(|\gamma(\sigma_{\Delta},\tilde{Z}^{k-1,f}_{i,x})|)}
\end{equation}
 \begin{equation}\label{pb_eq1}
    \tilde{p}_{b}(\theta_{k}) =
    \sum_{i =1}^{N_{b}^{k+1}}\sum_{x \in \mathbb{A}}\alpha^{k+1,b}_{i}\lambda^{k+1,b}_{i,x}\frac{e^{Re[\gamma(\sigma_{\Delta},\tilde{Z}^{k+1,b}_{i,x})e^{-j\theta_{k}}]}}{2\pi I_{0}(|\gamma(\sigma_{\Delta},\tilde{Z}^{k+1,b}_{i,x})|)}
\end{equation}

where
\begin{equation}\label{Z_f}
    \tilde{Z}^{k-1,f}_{i,x} = z^{k-1,f}_{i}+\frac{r_{k-1}x^{*}}{\sigma^2}
\end{equation}
\begin{equation}\label{coeff_f}
    \lambda^{k-1,f}_{i,x} = \frac{1}{A}P_{d}(c_{k-1}=x)\frac{I_{0}(|\tilde{Z}^{k-1,f}_{i,x}|)}{I_{0}(|z^{k-1,f}_{i}|)}
\end{equation}
\begin{equation}\label{Z_b}
    \tilde{Z}^{k+1,b}_{i,x} = z^{k+1,b}_{i}+\frac{r_{k+1}x^{*}}{\sigma^2}
\end{equation}
\begin{equation}\label{coeff_b}
    \lambda^{k+1,b}_{i,x} = \frac{1}{B}P_{d}(c_{k+1}=x)\frac{I_{0}(|\tilde{Z}^{k+1,b}_{i,x}|)}{I_{0}(|z^{k+1,b}_{i}|)}
\end{equation}
\begin{equation}\label{gamma}
    \gamma(\sigma_{\Delta},Z) = \frac{Z}{1+|Z|\sigma^{2}_{\Delta}}
\end{equation}
where $A$ and $B$ are a normalizing constants.

Therefore, equations (\ref{pf_eq1}) and (\ref{pb_eq1}) are Tikhonov mixtures of
order $N_{f}^{k}M$ and $N_{b}^{k}M$. Since we do not want to increase the mixture order every
symbol, a mixture dimension reduction algorithm must be derived which
captures ''most" of the information in the mixtures $\tilde{p}_{f}(\theta_{k})$ and
$\tilde{p}_{b}(\theta_{k})$, while keeping the computational complexity low. From
now on, we will present only the forward approximations, but the same
applies for the backward.

There are many metrics used for mixture reduction. The two most
commonly used are the Integral Squared Error (ISE) and the KL. The ISE
metric is defined for mixtures $f(\theta)$ and $g(\theta)$ as follows,

\begin{equation}\label{ise}
D_{ISE}(f(\theta)||g(\theta)) = \int_{0}^{2\pi}(f(\theta)-g(\theta))^{2}d\theta
\end{equation}

We chose the KL divergence for the cost function between the reduced
mixture and the original mixture rather than ISE, since the former is
expected to get better results. For example, assume a scenario where there is a low
probability isolated cluster of components, then if the reduction algorithm would prune that
cluster the ISE based cost will not be effected. However, the KL based reduction will have to assign a cluster
since the cost of not approximating it, is very high. In general, the KL divergence does not take in to account the probability of the
components while the ISE does. This feature of KL is useful since we wish to track all the significant
phase trajectories regardless of their probability.
We define the following mixture reduction task using the Kullback
Leibler divergence - \emph{Given a Tikhonov mixture $f(\theta)$ of
order $L$, find a Tikhonov mixture $g(\theta)$ of order $N$ ($L>N$),
which minimizes,}

\begin{equation}\label{obj}
D_{KL}(f(\theta)||g(\theta))
\end{equation}

\emph{where,}

\begin{equation}\label{orig_mix}
f(\theta)=  \sum_{i =1}^{L}\alpha_{i}f_{i}(\theta)
\end{equation}
\begin{equation}\label{red_mix}
g(\theta)=  \sum_{j =1}^{N}\beta_{j}g_{j}(\theta)
\end{equation}
where $f(\theta)$ is the mixture $\tilde{p}_{f}(\theta_{k})$ and the reduced order mixture $g(\theta)$
will be the next forward message, $p_{f}(\theta_{k})$.
We would like to provide an additional insight to choosing KL. The information theoretic meaning of KL
divergence is that we wish that the
loss in bits when compressing a source of probability $f(\theta)$,
with a code matched to the probability $g(\theta)$ will be not
larger than $\epsilon$. Thus, we wish to find a lower order mixture $f(\theta)$ which is a compressed
version of $f(\theta)$.

\subsection{Mixture Reduction algorithms - Review}
There is no analytical solution for (\ref{obj}), but there are many mixture reduction algorithms which
provide a suboptimal
solution for it. They can be generally classified in to two
groups, \emph{local} and \emph{global} algorithms. The
\emph{global} algorithms attempt to solve (\ref{obj}) by gradient
descent type solutions which are very computationally demanding. The \emph{local} algorithms usually start
from a large
mixture and prune out components/merge similar components, according
to some rule, until a target mixture order is reached. A very
good summary of many of these algorithms can be found in
\cite{sv2011}. The \emph{global} algorithms do not deal with KL
divergence and thus are not suited for our problem. We will review two
\emph{local} algorithms in the following section which provide the
best performance in the sense of best balancing the tradeoff between
complexity and accuracy, and show why they fail for our case.
The first algorithm is the one proposed in \cite{runnalls2007}. This
algorithm minimizes a \emph{local} problem, which \emph{sometimes} provides a
good approximation for (\ref{obj}).

Given (\ref{orig_mix}), the algorithm finds a pair of mixture components, $f_{i^{*}}$ and
$f_{k^{*}}$ which satisfy,

\begin{equation}\label{runnalls}
[i^{*},k^{*}] = arg\min_{i,k}D_{KL}(\alpha f_{i} + (1-\alpha)
f_{k}||g_{j}(\theta))
\end{equation}

where,
\begin{equation}\label{runnalls1}
g_{j}(\theta) = CMVM(\alpha f_{i} + (1-\alpha) f_{k})
\end{equation}

and $\alpha$ is normalized probability of $f_{i}$ after dividing by the
sum of the probabilities of $f_{i}$ and $f_{k}$.
The algorithm merges the two components to $g_{j}(\theta)$, thus the order of (\ref{orig_mix}) has now
decreased by one. This procedure is now repeated on the new mixture iteratively
to find another optimal pair until the target mixture order is reached.
It should be noted that the component's probability influences
the metric (\ref{runnalls}). Suppose we have two very different components, one with
high probability and another with very low probability, which is the correct hypothesis. Then the
algorithm may choose to
cluster them, and the low probability component will be lost which may be the correct trajectory.
Another algorithm is the one proposed in \cite{goldberger2004hierarchical}, which also does not directly
solve
(\ref{obj}), but defines another metric which is much easier to handle
mathematically. The algorithm's operation is very similar to the
K-means algorithm. It first chooses an initial reduced mixture
$g(\theta)$ and then iteratively performs the following,

\begin{enumerate}
  \item \emph{Select the clusters} - Map all $f_{i}$ to the $g_{j}$
which minimizes $D_{KL}(f_{i}||g_{j})$
  \item \emph{Regroup} - For all $j$, optimally cluster the elements
$f_{i}$ which were mapped to each $g_{j}$ to create the new
$g(\theta)$
\end{enumerate}

This algorithm is dependent on initial conditions in order to converge to the lowest mixture. Also, the
iterative
process increases the computational complexity significantly.
In \cite{goldberger2004hierarchical} and \cite{runnalls2007}, the
Gaussian case was considered, thus the clustering was performed using
Gaussian moment matching. For our setting, we have taken the liberty to change the moment
matching to CMVM, since we have Tikhonov distributions and not Gaussian. Note that in both algorithms, the
target order must
be defined before operation, since they have to know when to stop.
Selecting the proper target mixture order is a difficult task.
On one hand, if we choose a large target order, then the complexity
will be too high. On the other hand, if we choose the order to be low
then the algorithm may cluster components which clearly need not be
merged but since they provide the minimal KL divergence, they are
clustered.
Therefore, in order to maintain a good level of accuracy, the task
should be to guarantee an upper bound on the KL divergence and not try
to unsuccessfully minimize it. Moreover, it should be noted that in
our setting the mixture reduction task (\ref{obj}), is performed many
times and not once. Therefore, there may not be a need to have the
same reduced mixture order for each symbol. These ideas will lead us to
the approach presented in the next section of the \emph{adaptive}
mixture canonical model.

\section{A New Approach to Mixture Reduction}
We have seen that the current state of the art low complexity mixture
reduction algorithms based on a fixed target mixture order do not provide good enough approximations to
(\ref{obj}). Moreover, the choice of the mixture order plays a crucial part in the clustering task. On one
hand, a small mixture will provide poor
SP message approximation which will propagate over the factor graph
and cause a degradation in performance. On the other hand, a large
mixture order will demand too many computational resources.
Instead of reducing (\ref{pf_eq1}) and (\ref{pb_eq1}) to a fixed order,
we propose a new approach which has better accuracy while keeping low
complexity. Since we are performing Bayesian inference on a large data
block, we have many mixture reductions to perform rather than just a single reduction. Therefore, in terms
of
computational complexity, it is useful to use different mixture orders
for different symbols and look at the average number of components as
a measure of complexity. This new observation is
critical in achieving high accuracy and low PER while keeping
computational complexity low.
We define the \emph{new mixture reduction task} -
\emph{Given a Tikhonov mixture $f(\theta)$,}
\begin{equation}\label{in}
f(\theta) = \sum_{i=1}^{L}\alpha_{i}f_{i}(\theta)
\end{equation}
\emph{Find the Tikhonov mixture $g(\theta)$ with the minimum number of
components $N$}

\begin{equation}\label{out}
g(\theta) = \sum_{j=1}^{N}\beta_{j}g_{j}(\theta)
\end{equation}

\emph{which satisfy,}

\begin{equation}\label{task}
D_{KL}(f(\theta)||g(\theta)) \leq \epsilon
\end{equation}

Solving this new task will guarantee that the accuracy of the
approximation is upper bounded so we can keep the PER
levels low. Moreover, simulations show that the resulting mixtures are of very small sizes. In the
following section, we will show a low complexity algorithm which finds a mixture $g(\theta)$ whose average
number of mixture components is low.

\subsection{Mixture Reduction Algorithm}
In this section, a mixture reduction algorithm is proposed which is suboptimal in
the sense that it does not have the minimal number of components, but
finds a low order mixture which satisfies (\ref{task}), for any $\epsilon$. The
algorithm, whose details are given in pseudo-code in Algorithm 1, uses
the CMVM approach, for optimally merging a Tikhonov mixture to a
single Tikhonov distribution.

\begin{algorithm}
\caption{Mixture Reduction Algorithm}
\label{mix_redc_algo}
{\fontsize{10}{5}
\begin{algorithmic}
\State $j \gets 1$
\While{$|f(\theta)| > 0$}
\State $lead \gets argmax\{\underline{\alpha}\}$
 \For{$i = 1 \to |f(\theta)|$}
     \If {$D_{KL}(f_{i}(\theta) || f_{lead}(\theta)) \leq \epsilon $}
         \State $idx \gets [idx , i]$
    \EndIf
\EndFor
\State $\beta_{j} \gets \sum_{i\in idx}\alpha_{i}$
\State $g_{j}(\theta) \gets CMVM(\sum_{i\in
idx}\frac{\alpha_{i}}{\beta_{j}}f_{i}(\theta))$
\State $f(\theta) \gets f(\theta) - \sum_{i \in idx}{\alpha_{i}f_{i}(\theta)} $
\State Normalize $f(\theta)$
\State $j \gets j + 1$
\EndWhile
\end{algorithmic}}
\end{algorithm}

The input to this algorithm, $f(\theta)$, is the Tikhonov mixture (\ref{pf_eq1}) and the output Tikhonov
mixture $g(\theta)$ is a reduced version of
$f(\theta)$ and approximates the next forward or backward messages. Note that the function $|f(\theta)|$
outputs the number of Tikhonov components in the Tikhonov mixture $f(\theta)$. The computations of $D_{KL}(f_{i}(\theta) || f_{lead}(\theta))$ and $CMVM(\sum_{i\in idx}\frac{\alpha_{i}}{\beta_{j}}f_{i}(\theta))$ are detailed in appendices (\ref{sec:kl_comp}) and (\ref{sec:cmvm_exmp}).
In the beginning of each iteration, the algorithm selects the highest probability mixture component and
clusters it with all the components which are similar to it (KL sense). It then finds the next highest
probability component and performs the same until there are no components left to cluster. We will now
show that for any $\epsilon$, the algorithm satisfies (\ref{task}).

\begin{theorem}
\emph{(Mixture Reduction Accuracy):}
\label{mix_red_acc}
Let $f(\theta)$ be a Tikhonov mixture of order $L$ and $\epsilon$ be a real positive number. Then,
applying the Mixture Reduction Algorithm 1 to $f(\theta)$ using $\epsilon$, produces a Tikhonov mixture
$g(\theta)$, of order $N$ which satisfies,
\begin{equation}\label{mix_red_thr}
     D_{KL}(f(\theta)||g(\theta)) \leq \epsilon
\end{equation}
\end{theorem}
\begin{proof}
In the first iteration, the algorithm selects the highest probability mixture component of (\ref{in}) and
denotes it as $f_{lead}(\theta)$.
Let $M_{0}$, be the set of mixture components $f_{i}(\theta)$ selected for clustering,
\begin{equation}\label{rule_kl}
  M_{0} = \{f_{i}(\theta) \: | \: D_{KL}(f_{i}(\theta) || f_{lead}(\theta)) \leq \epsilon \}
\end{equation}
and $M_{1}$ be the set of mixture components which were not selected,
\begin{equation}\label{rule_kl}
  M_{1} = \{f_{i}(\theta) \: | \: D_{KL}(f_{i}(\theta) || f_{lead}(\theta)) > \epsilon \}
\end{equation}
Thus,
\begin{equation}\label{16}
   \sum_{i \in M_{0}}\frac{\alpha_{i}}{\beta_{1}}D_{KL}(f_{i}(\theta)||f_{lead}(\theta))
\leq \epsilon
\end{equation}
where,
\begin{equation}\label{15}
   \beta_{1} = \sum_{i \in M_{0}}\alpha_{i}
\end{equation}
Using Lemma (\ref{kl_bound1}),
\begin{equation}\label{new_KL2}
    D_{KL}\left(\sum_{i \in
M_{0}}\frac{\alpha_{i}}{\beta_{1}}f_{i}(\theta)||f_{lead}(\theta)\right) \leq
\epsilon
\end{equation}
The algorithm then clusters all the distributions in $M_{0}$ using CMVM,
\begin{equation}\label{17}
   g_{1}(\theta) =
CMVM\left(\sum_{i \in M_{0}}\frac{\alpha_{i}}{\beta_{1}}f_{i}(\theta)\right)
\end{equation}
then, using Theorem (\ref{mix_tikh_thr}),
\begin{multline}\label{new_KL2}
    D_{KL}\left(\sum_{i \in
M_{0}}\frac{\alpha_{i}}{\beta_{1}}f_{i}(\theta)||g_{1}(\theta)\right) \leq \\
D_{KL}\left(\sum_{i \in
M_{0}}\frac{\alpha_{i}}{\beta_{1}}f_{i}(\theta)||f_{lead}(\theta)\right)
\end{multline}
which means that,
\begin{equation}\label{new_KL3}
    D_{KL}\left(\sum_{i \in
M_{0}}\frac{\alpha_{i}}{\beta_{1}}f_{i}(\theta)||g_{1}(\theta)\right) \leq
\epsilon
\end{equation}
We can rewrite the mixtures $f(\theta)$ and $g(\theta)$ in the following way,
\begin{equation}\label{new_f}
    f(\theta) = \alpha_{M_{0}}f_{M_{0}}(\theta) + \alpha_{M_{1}}f_{M_{1}}(\theta)
\end{equation}
\begin{equation}\label{new_g}
    g(\theta) = \beta_{1}g_{1}(\theta) + (1-\beta_{1})h(\theta)
\end{equation}
where,
\begin{equation}\label{thr_2}
    \alpha_{M_{0}} = \sum_{i \in M_{0}}\alpha_{i}
\end{equation}
\begin{equation}\label{thr_2_1}
    \alpha_{M_{1}} = \sum_{i \in M_{1}}\alpha_{i}
\end{equation}
\begin{equation}\label{thr_3}
    f_{M_{i}}(\theta) = \sum_{j \in M_{i}}\frac{\alpha_{j}}{\alpha_{M_{i}}}f_{j}(\theta)
\end{equation}
Using (\ref{15}),
\begin{equation}\label{thr_3}
    \alpha_{M_{i}} = \beta_{i}
\end{equation}
Therefore (\ref{new_f}) and (\ref{new_g}) are two mixtures of the same size and have exactly the same
coefficients, thus the KL of the probability mass functions induced by the coefficients of both mixtures
is zero. Using Lemma (\ref{kl_bound3}),
\begin{multline}\label{new_KL4}
   D_{KL}(f(\theta) || g(\theta)) \leq \beta_{1}D_{KL}(f_{M_{0}}(\theta) || g_{1}(\theta) \\ + (1 -\beta_{1})D_{KL}(f_{M_{1}}(\theta) || h(\theta))
\end{multline}
using (\ref{new_KL2}) we get,
\begin{multline}\label{new_KL5}
   D_{KL}(f(\theta) || g(\theta)) \leq \beta_{1}\epsilon \\ + (1 - \beta_{1})D_{KL}(f_{M_{1}}(\theta) ||h(\theta))
\end{multline}
If we find a Tikhonov mixture $h(\theta)$ ,which satisfies,
\begin{equation}\label{new_KL7}
D_{KL}(f_{M_{1}}(\theta) || h(\theta)) \leq \epsilon
\end{equation}
then we will prove the theorem. But (\ref{new_KL7}) is exactly the same as the original problem, thus
applying the same clustering steps as described earlier on the new mixture $f_{M_{1}}(\theta)$ will
ultimately satisfy,
\begin{equation}\label{new_KL6}
   D_{KL}(f(\theta) || g(\theta)) \leq \epsilon
\end{equation}
\end{proof}

\subsection{Mixture Reduction As Phase Noise Tracking}
\label{sec:phase_tracking}
Recall in Fig. \ref{fig:splits}, that the phase noise messages can be
viewed as multiple separate phase trajectories, then the mixture reduction
algorithm can be viewed as a scheme to map the different mixture components to
different phase trajectories. The mixture reduction algorithm receives
a mixture describing the next step of all the trajectories and assigns it to a specific trajectory, thus
we are
able to accurately track all the hypotheses for all the phase trajectories.
Assuming slowly varying phase noise and high SNR, the mixture reduction tracking loop $i$, $\hat{\theta}^{i}_{k}$ for each trajectory can be computed in the following manner,
\begin{equation}\label{pll_equiv_op}
    \hat{\theta}^{i}_{k} =  \hat{\theta}^{i}_{k-1} +
\frac{|r_{k-1}||c_{t}|}{G_{k-1}\sigma^2}(\angle{r_{k-1}}+\angle{c_{t}}-\hat{\theta}^{i}_{k-1})
\end{equation}
where, $c_{t}$ and $G_{k-1}$ are a soft decision of the constellation
symbol and the inverse conditional MSE for $\hat{\theta}_{k-1}$,
respectively. The proof for this claim is provided in appendix \ref{sec:pll}.
Thus the mixture reduction is equivalent to multiple soft decision first order PLLs with adaptive loop
gains. Whenever the mixture components of the SPA message become too far apart, a split occurs and
automatically the number of tracking loops increases in order to track the new trajectories.

\subsection{Limited Order Adaptive Mixture}
In the previous section, we have presented an algorithm which
\emph{adaptively} changes the canonical model's mixture order, with no
upper bound. This enabled us to track \emph{all} the significant phase
trajectories in the SP messages. However, there may be scenarios with limited complexity, in which we are
forced
to have a limited number of mixture components, thus we can track only a limited number of phase
trajectories. If the number of significant phase trajectories is larger than the
maximum number of mixture components allowed, then we might miss the correct
trajectory. For example, if we limit the number of tracked trajectories to one, we get an algorithm very
close to a PLL. In this case whenever a split event occurs, we have to choose one of the trajectories and
abandon the other and in case we chose the wrong one, we experience a cycle slip. Analogously we can call
cycle slip the event of missing the right trajectory even when more than one trajectory is available.
In this section, assuming pilots are present, we propose an improvement to Algorithm 1, which provides a
solution to the missed trajectories problem. The improved algorithm still uses a mixture canonical
model for the approximation of messages in the SPA but with an additional
variable $\phi^{f}_{k}$ (for backward recursions $\phi^{b}_{k}$ ), which approximates, online, the
probability that
the tracked trajectories include the correct one. This approach enables us to track phase
trajectories while maintaining a level of their confidence. We apply the previously used clustering based
on the KL
divergence in order to select which of the components of the mixture
are going to be approximated by a Tikhonov mixture, while the rest of
the components will be ignored, but their total probabilities will be accumulated. We then use pilot
symbols and $\phi^{f}_{k}$ in order to regain tracking if a cycle slip has occurred. This approach proves
to be robust to
phase slips and provides a high level of accuracy while keeping a low
computational load. The resulting algorithm was shown, in simulations,
to provide very good performance in high phase noise level and very
close to the performance of the optimal algorithm even for
mixtures of order 1,2 and 3.
\newline

\subsubsection{Modified Reduction Algorithm}
We denote the modification of Algorithm 1 for limited complexity, as Algorithm 2. This algorithm selects
\textbf{some} components from a Tikhonov mixture, $f(\theta)$ and clusters them to an output Tikhonov
mixture $g(\theta)$ of maximum order $L$.
\begin{algorithm}
\caption{Modified Mixture Reduction Algorithm}
\label{mix_redc_algo}
{\fontsize{10}{5}
\begin{algorithmic}
\State $j \gets 1$
\While{$j \leq L$  or  $|f(\theta)| > 0$}
\State $lead \gets argmax\{\underline{\alpha}\}$
 \For{$i = 1 \to |f(\theta)|$}
     \If {$D_{KL}(f_{i}(\theta) || f_{lead}(\theta)) \leq \epsilon $}
         \State $idx \gets [idx , i]$
    \EndIf
\EndFor
\State $\beta_{j} \gets \sum_{i\in idx}\alpha_{i}$
\State $g_{j}(\theta) \gets CMVM(\sum_{i\in
idx}\frac{\alpha_{i}}{\beta_{j}}f_{i}(\theta))$
\State $f(\theta) \gets f(\theta) - \sum_{i \in idx}{\alpha_{i}f_{i}(\theta)} $
\State Normalize $f(\theta)$
\State $j \gets j + 1$
\EndWhile
\State $\phi^{f}_{k} \gets (\sum_{j}\beta_{j})\phi^{f}_{k-1}$
\end{algorithmic}}
\end{algorithm}
We initialize $\phi^{f}_{0} = 1$, which means that in the first received sample, for the forward
recursion, there is no cycle slip. Note that Algorithm 2, is identical to Algorithm 1 apart for the
computation of $\phi^{f}_{k}$.
For each iteration, Algorithm 2, selects the most probable component in (\ref{pf_eq1}) and clusters all
the mixture components similar to it. The algorithm then removes this cluster and finds another cluster
similarly. When there are no more components in $f(\theta)$ or the maximum allowed mixture order is
reached, the algorithm computes $\phi^{f}_{k}$. As discussed earlier, this variable represents the
probability that a cycle slip has not occurred. The algorithm sums up the probabilities of the clustered
components in $f(\theta)$ and multiplies that with $\phi^{f}_{k-1}$ to get $\phi^{f}_{k}$. Suppose we have
clustered all the components in $f(\theta)$, then $\phi^{f}_{k-1}$ will be equal to  $\phi^{f}_{k}$. That
suggests that the probability that a cycle slip has occurred before sample $k-1$ is the same as for sample
$k$. This is in agreement with the fact that no trajectories were ignored at the reduction from $k-1$ to
$k$. For low enough $\epsilon$ , $\phi^{f}_{k}$ is a good approximation of that probability.
\newline

\subsubsection{Recovering From Cycle Slips}
In this section, we propose to use $\phi^{f}_{k-1}$, the probability that a cycle has not occurred, and
the information conveyed by pilots in order to combat cycle slips. In case of a cycle slip, the phase
message estimator based on the tracked trajectories is useless and we need to find a better estimation of
the phase message. We propose to estimate the message using \textbf{only} the pilot symbol, $p_{d}(\theta_{k-1})$. However, if a cycle slip has not occurred, then estimating the phase message based
\textbf{only} on the pilot symbol might damage our tracking. Therefore, once a pilot symbol arrives, we
will \textbf{average} the two proposed estimators according to $\phi^{f}_{k-1}$,
\begin{equation}\label{avg_cycle}
    q_{f}(\theta_{k-1}) = \phi^{f}_{k-1}p_{f}(\theta_{k-1})+(1-\phi^{f}_{k-1})\frac{1}{2\pi}
\end{equation}

If a cycle slip has occurred and $\phi^{f}_{k-1}$ is low, then the pilot will, in high probability,
correct the tracking. We present the proposed approach in pesudo-code in Algorithm (\ref{cycle_slip_algo}).

\begin{algorithm}
\caption{Forward Message Computation with Cycle Slip Recovery}
\label{cycle_slip_algo}
{\fontsize{10}{5}
\begin{algorithmic}
\State $p_{f}(\theta_{0}) \gets \frac{1}{2\pi}$
\State $\phi^{f}_{0} \gets 1$
\State $k \gets 1$
\While{$k \leq K$}
    \State Compute $p_{d}(\theta_{k-1})$
     \If {$c_{k-1}$ is a pilot}
         \State $ q_{f}(\theta_{k-1}) \gets \phi^{f}_{k-1}p_{f}(\theta_{k-1})+(1-\phi^{f}_{k-1})\frac{1}{2\pi}$
         \State $t \gets 1$
     \Else
        \State $ q_{f}(\theta_{k-1}) \gets  p_{f}(\theta_{k-1}) $
        \State $t \gets \phi^{f}_{k-1}$
     \EndIf
     \State $\tilde{p}_{f}(\theta_{k}) \gets \int_{0}^{2\pi}q_{f}(\theta_{k-1})p_{d}(\theta_{k-1})p_{\Delta}(\theta_{k}-\theta_{k-1})d\theta_{k-1}$
     \State $[p_{f}(\theta_{k}),\phi^{f}_{k}] \gets Algorithm2(\tilde{p}_{f}(\theta_{k}),t)$
     \State $k \gets k + 1$
\EndWhile
\end{algorithmic}}
\end{algorithm}

\section{Computation of $P_{u}(c_{k})$}
As discussed in section (\ref{sys_model}), after computing the forward and backward messages, the next
step of the SP algorithm is to compute $P_{u}(c_{k})$. These messages describe the LLR of a code symbol
based on the channel part of the factor graph. These messages are sent to the LDPC decoder and the correct
approximation of these messages is crucial for the decoding of
the LDPC.
When using Algorithm 1 for the computation of the forward and backward messages, we use the reduced
mixtures with (\ref{Pu}) and analytically compute the message. However, when using a limited order mixture
and Algorithm 2 with the cycle slip recovery method in Algorithm 3, we use $\phi^{f}_{k}$ and $\phi^{b}_{k}$ in order to better the estimation of the messages. Thus $P_{u}(c_{k})$ is a weighted summation of
four components which can be interpreted as conditioning on the probability that a phase
slip has occurred for each recursion (forward and backward). This will
ensure that the computation of $P_{u}(c_{k})$ is based on the most
reliable phase posterior estimations, even if a phase slip has
occurred in a single recursion (forward or backward).
We insert the mixture (\ref{avg_cycle}) into (\ref{Pu}),
\begin{equation}\label{Pu_new1}
    P_{u}(c_{k}) \propto \int_{0}^{2\pi}q_{f}(\theta_{k})q_{b}(\theta_{k})e_{k}(c_{k},\theta_{k})d\theta_{k}
\end{equation}

where $q_{f}(\theta_{k})$ and $q_{b}(\theta_{k})$ are defined in Algorithm 3.
We decompose the computation to a summation of four components,
\begin{equation}\label{Pu_new}
    P_{u}(c_{k}) \propto A + B + C + D
\end{equation}
where
\begin{equation}\label{A}
    A = \phi^{f}_{k}\phi^{b}_{k}\int_{0}^{2\pi}p_{f}(\theta_{k})p_{b}(\theta_{k})e_{k}(c_{k},\theta_{k})d\theta_{k}
\end{equation}
\begin{equation}\label{B}
    B = \phi^{f}_{k}(1-\phi^{b}_{k})\int_{0}^{2\pi}p_{f}(\theta_{k})\frac{1}{2\pi}e_{k}(c_{k},\theta_{k})d\theta_{k}
\end{equation}
\begin{equation}\label{C}
    C = (1-\phi^{f}_{k})\phi^{b}_{k}\int_{0}^{2\pi}\frac{1}{2\pi}p_{b}(\theta_{k})e_{k}(c_{k},\theta_{k})d\theta_{k}
\end{equation}
\begin{equation}\label{D}
    D = (1-\phi^{b}_{f})(1-\phi^{b}_{k})\int_{0}^{2\pi}\frac{1}{2\pi}\frac{1}{2\pi}e_{k}(c_{k},\theta_{k})d\theta_{k}
\end{equation}
We will detail the computation of $A$, but the same applies to the
other components of (\ref{Pu_new}). We use the mixture form defined in (\ref{new_pf}) and (\ref{new_pb}).

We define the following,
\begin{equation}\label{Z_mix}
    Z_{\psi} = z^{k,f}_{i}+z^{k,b}_{j}+\frac{r_{k}c^{*}_{k}}{\sigma^{2}}
\end{equation}
and get,
\begin{equation}\label{19}
   A = \sum_{i=1}^{N_{f}^{k}}\sum_{j=1}^{N_{b}^{k}}\alpha^{k,f}_{i}\alpha^{k,b}_{j}\frac{I_{0}(|Z_{\psi}|)}{2\pi
I_{0}(|z^{k,f}_{i}|)I_{0}(|z^{k,b}_{j}|)}
\end{equation}

When implementing the algorithm in log domain, we can simplify (\ref{19}), by using (\ref{assumps3}),
\begin{multline}\label{20}
   \log\left(\frac{I_{0}(|Z_{\psi}|)}{2\pi
I_{0}(|z^{k,f}_{i}|)I_{0}(|z^{k,b}_{j}|)}\right) \approx |Z_{\psi}| - |z^{k,f}_{i}| - |z^{k,b}_{j}| \\ -\frac{1}{2}\log\left(\frac{|Z_{\psi}|}{|z^{k,f}_{i}||z^{k,b}_{j}|}\right)
\end{multline}
and for large enough $|z^{k,f}_{i}|$ and $|z^{k,b}_{j}|$
\begin{equation}\label{21}
   \log\left(\frac{I_{0}(|Z_{\psi}|)}{2\pi
I_{0}(|z^{k,f}_{i}|)I_{0}(|z^{k,b}_{j}|)}\right) \approx |Z_{\psi}| - |z^{k,f}_{i}| - |z^{k,b}_{j}|
\end{equation}

\section{Complexity}
In this section we will detail the computational complexity of the proposed algorithms and compare the
complexity to the DP and BARB algorithms. Since the mixture order changes between symbols and LDPC
iterations, we can not give an exact expression for the computational complexity. Therefore, in order to
assess the complexity of the algorithms, we denote the average number of components in the canonical model
per sample, as $\gamma(i)$, where $i$ is the index of the LDPC iteration. $\gamma(i)$, decreases in
consecutive LDPC iterations due to the fact that the LDPC decoder provides better soft information on the
symbols thus resolving ambiguities and decreasing the required number of components in the mixture. This
value, $\gamma(i)$, depends mainly on the number of ambiguities that the phase estimation algorithm
suffers. These ambiguities are a function of the SNR, phase noise variance and algorithmic design
parameters such as the number of LDPC iteration, KL threshold - $\epsilon$ and the pilot pattern.

The significant difference in computational complexity between the DP and the mixture based algorithms
stems from the fact that multi modal SPA messages are not well characterized by a single Tikhonov and the
DP algorithm must use many quantization levels to accurately describe them. However, the mixture algorithm
is successful in characterizing these messages using few mixture parameters and this difference is very
significant as the modulation order increases.
The mixture algorithm starts out by approximating the forward and backward messages using Tikhonov
mixtures. These mixtures are then inserted in to (\ref{pf}) and (\ref{pb}) to produce larger mixtures
(\ref{pf_eq1}) and (\ref{pb_eq1}). Next, the mixture reduction scheme produces a reduced mixture which is
used to compute $P_{u}(c_{k})$. On average, for a given LDPC iteration $i$, the forward message, $p_{f}(\theta_{k})$, is a Tikhonov mixture of order $\gamma(i)$. After applying (\ref{pf}), the mixture
increases to order $M\gamma(i)$ and is sent to the mixture reduction algorithm. Also on average, the
clustering algorithm performs $\gamma(i)$ clustering operations on $M$ components. The clustered mixtures
are then used to compute $P_{u}(c_{k})$ which is a multiplication of the forward and backward mixtures.
In appendices (\ref{sec:cmvm_exmp}) and (\ref{sec:kl_comp}), we have described the computation of the KL
divergence, $D_{KL}(f_{i}(\theta) || f_{lead}(\theta))$ and the application of the CMVM operator on the
clustered components - $g_{j}(\theta) \gets CMVM(\sum_{i\in idx}\frac{\alpha_{i}}{\beta_{j}}f_{i}(\theta))$.
In order to further reduce the complexity of the proposed algorithm, the variables representing
probabilities are stored in log domain and summation of these variables is approximated using the $\max$
operation. We also use the fact that for large $x$, $\log(I_{0}(x)) \approx x$ and approximate the KL
divergence in (\ref{kl_comp_full}) as,
\begin{equation}\label{kl_comp_approx}
  D_{KL} \approx |z_{2}|(1-cos(\angle z_{1} - \angle z_{2}))
\end{equation}
There is an option to abandon the clustering altogether, and replace it by a component selection algorithm,
which maintains the specified accuracy but requires more components in return. Now the complexity of
clustering is traded against the complexity of other tasks. The selection algorithm is a simple
modification in the algorithm. Instead of using CMVM to cluster several close components, we simply choose
$f_{lead}(\theta)$ as the result of the clustering.
Recalling (\ref{new_KL2}), we note that $f_{lead}(\theta)$ satisfies the accuracy condition and Theorem
\ref{mix_red_acc} still holds. Thus we will not suffer degradation in maximum error if we use this approximation
and not CMVM. However, the mean number of mixture components will increase since we do not perform any
clustering. The CMVM operator actually reduces the KL divergence between the original mixture and the
reduced mixture to much less than $\epsilon$. Therefore, when using CMVM, the reduced mixture is much
smaller than needed to satisfy the accuracy condition. In order to get the same performance with the
reduced algorithm, we need to decrease $\epsilon$ and use more components.
The reduced complexity is summarized in Table \ref{tab:model_complexity_reduced}, and compared to DP and
BARB. $Q$ is the number of quantization levels per constellation symbol in the DP algorithm. We only count
multiplication and LUT operations since they are more costly than additions. We assume that the cosine
operation is implemented using a look up table.
\begin{table*}[!ht]
\caption{Computational load per code symbol per iteration for M-PSK
constellation}  

\centering  
\begin{tabular}{p{1cm}|p{4cm}|p{2cm}|p{3cm}}                       
  &DP & BARB & Limited Order
\\
\hline\hline             
MULS & $4Q^{2}M^{2}+2M^{2}Q + 6MQ + M$ & $7M + 5$ & $4M\gamma(i)^{2}+2M(\gamma(i)+1)$\\
& & &\\
\hline
LUT & $QM$ & $3M$ &$3M\gamma(i)^{2}-\gamma(i)(2M-1)$\\
& & &\\
\hline                          
\end{tabular}
\label{tab:model_complexity_reduced}
\end{table*}

\section{Numerical Results}
In this section, we analyze the performance of the algorithms proposed in this paper. The performance
metrics of a decoding scheme is comprised of two parameters - the Packet/Bit Error Rate (PER/BER) and the
computational complexity. We use the DP algorithm as a benchmark for the lowest achievable PER and the
algorithm proposed in
\cite{barb2005}, denoted before as BARB as a benchmark for a state of the art low complexity scheme. The
phase noise model used in all the simulations is a Wiener process and the DP algorithm was simulated using
16 quantization levels between two constellation points. Also, note that the simulation results presented
in this paper use an MPSK constellation but the algorithm can also be applied, with small changes, to QAM
or any other constellation.

In Fig. \ref{fig:ber8psk} and \ref{fig:per8psk}, we show the BER and PER results for an 8PSK constellation
with an LDPC code of length 4608 with code rate 0.89. We chose $\sigma_{\Delta}=0.05$[rads/symbol] and a
single pilot was inserted every 20 symbols.

\begin{figure}
  \centering
  \includegraphics[width=8.4cm]{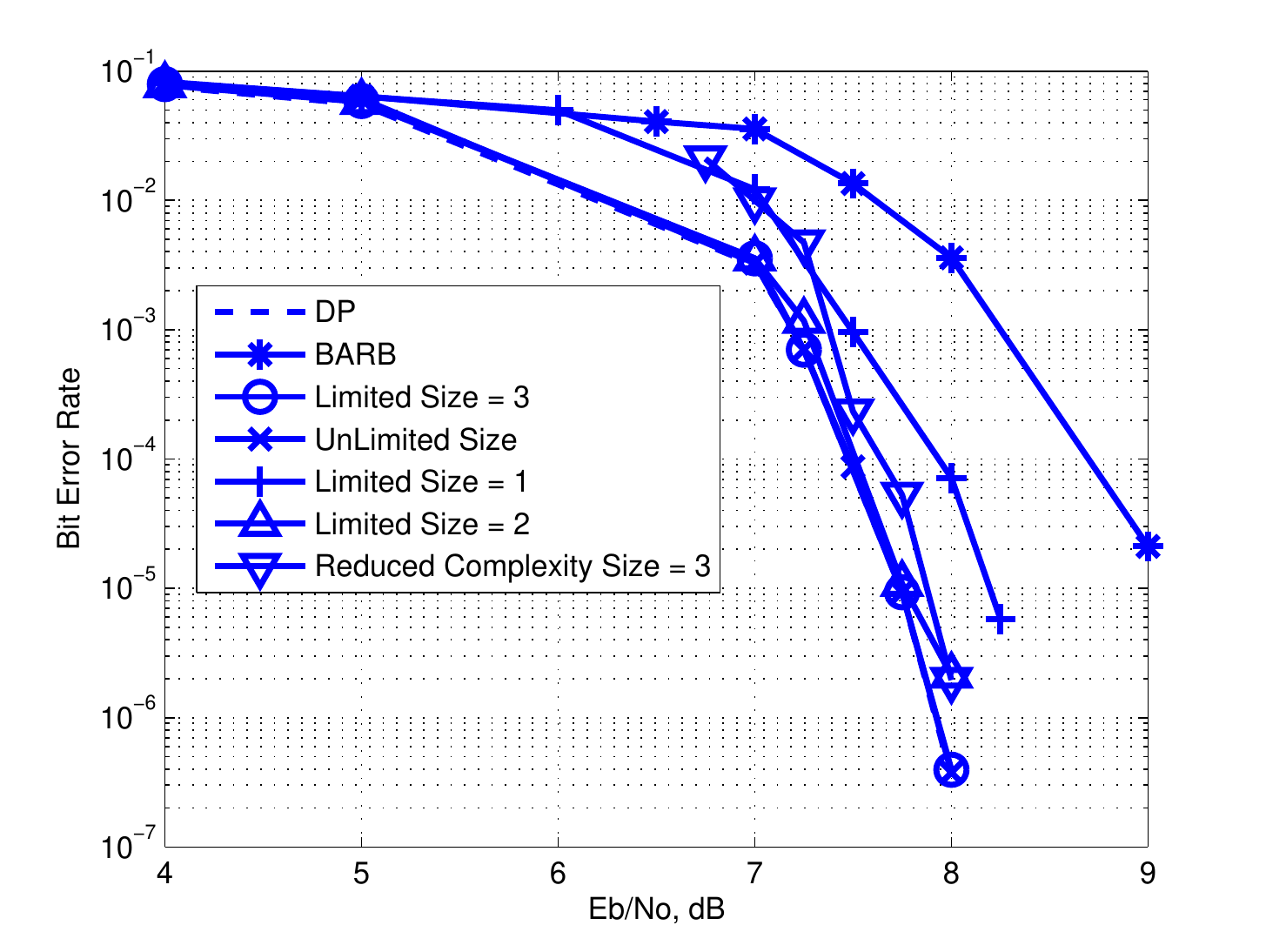}\\
  \caption{Bit Error Rate error rate - 8PSK , $\sigma_{\Delta}=0.05$, Pilot Frequency $= 0.05$}\label{fig:ber8psk}
\end{figure}

\begin{figure}
  \centering
  \includegraphics[width=8.4cm]{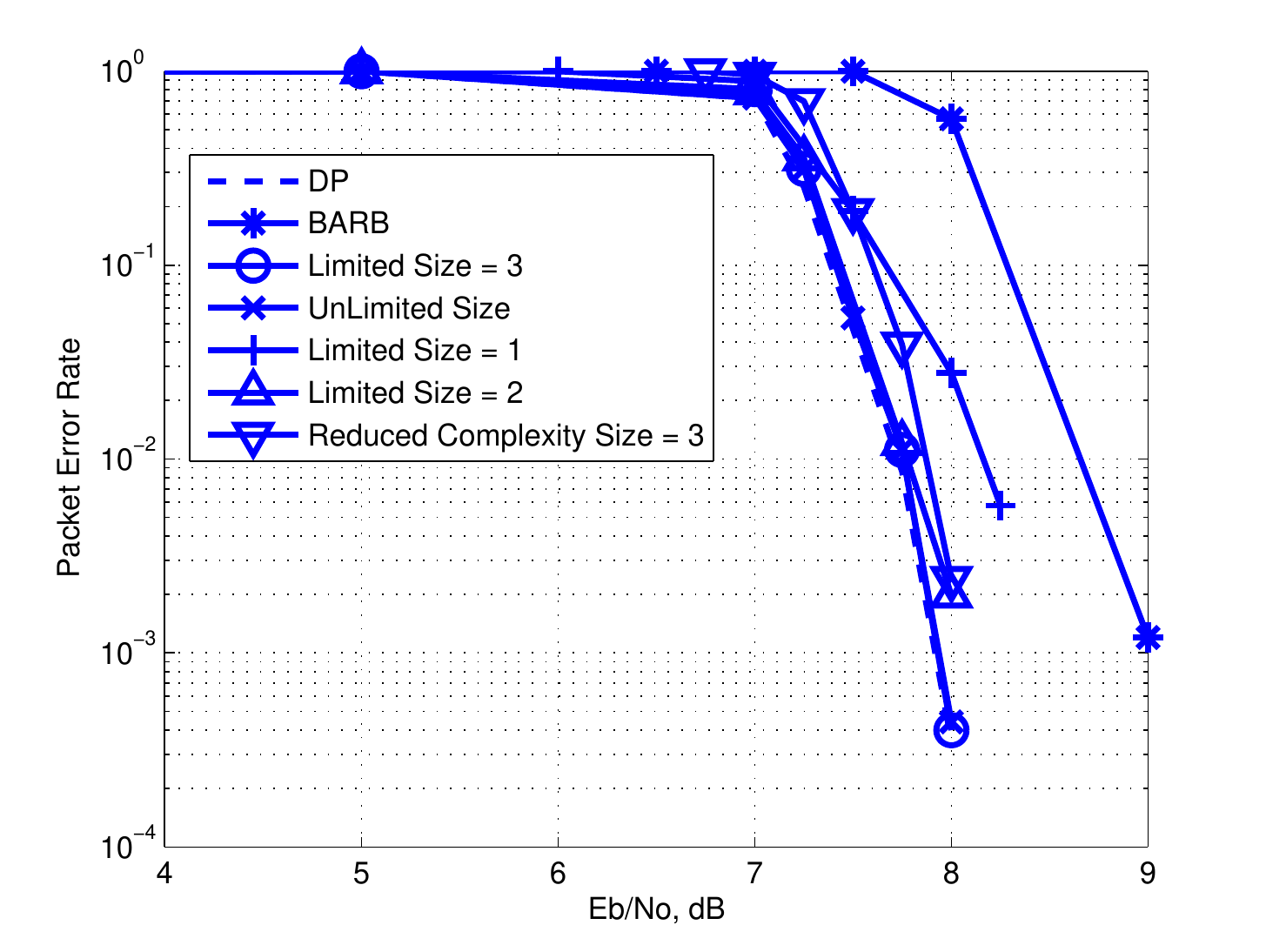}\\
  \caption{Packet Error Rate error rate - 8PSK , $\sigma_{\Delta}=0.05$, Pilot Frequency $= 0.05$}\label{fig:per8psk}
\end{figure}

The algorithms simulated were the unlimited order algorithm, the limited order algorithm with varying
mixture orders (1,2 and 3) and the reduced complexity algorithm of Order 3 (denoted Reduced Complexity Size 3). We
can see that the unlimited mixture, the limited order mixtures of order 2 and 3 and the reduced complexity
algorithm provide almost identical results, which are close to the
performance of the DP algorithm. On the other hand, the BARB algorithm has significant degradation with
respect to all the algorithms. We note that a mixture with only one component can not describe the phase
trajectory well enough to have PER levels like DP, but this algorithm is still better than BARB.\\
In Figs. \ref{fig:perbpsk},\ref{fig:perqpsk} and \ref{fig:per32psk} we show the PER results for a
BPSK,QPSK and 32PSK constellations respectively with the same code used earlier. For the BPSK and QPSK
scenarios we simulated the phase noise using $\sigma_{\Delta}=0.1$[rads/symbol] and for 32PSK we used $\sigma_{\Delta}=0.01$[rads/symbol]. A single pilot was inserted according to the pilot frequency detailed
in each figure's caption.

\begin{figure}
  \centering
  \includegraphics[width=8.4cm]{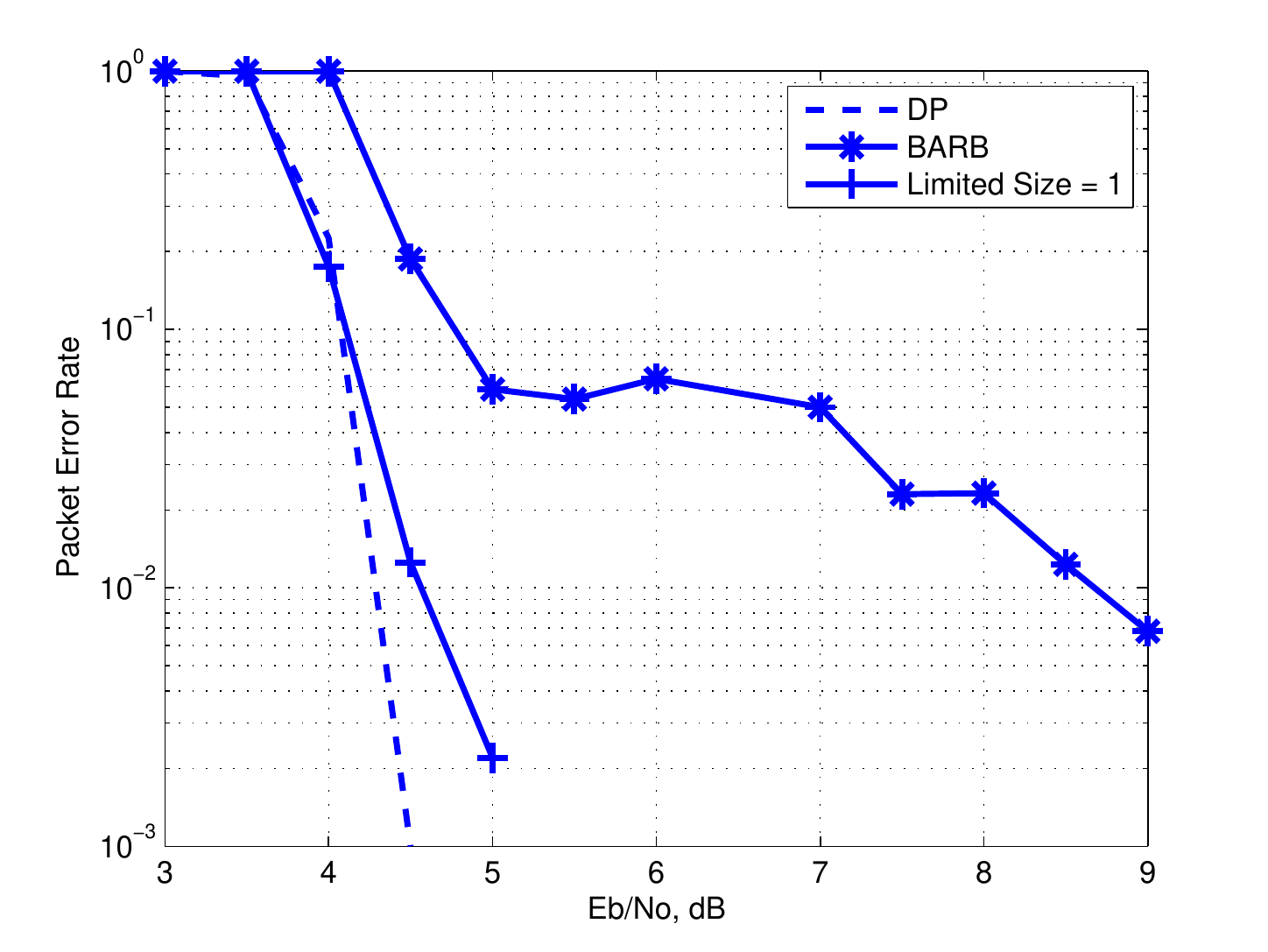}\\
  \caption{Packet Error Rate error rate - BPSK , $\sigma_{\Delta}=0.1$, Pilot Frequency $= 0.0125$}\label{fig:perbpsk}
\end{figure}

\begin{figure}
  \centering
  \includegraphics[width=8.4cm]{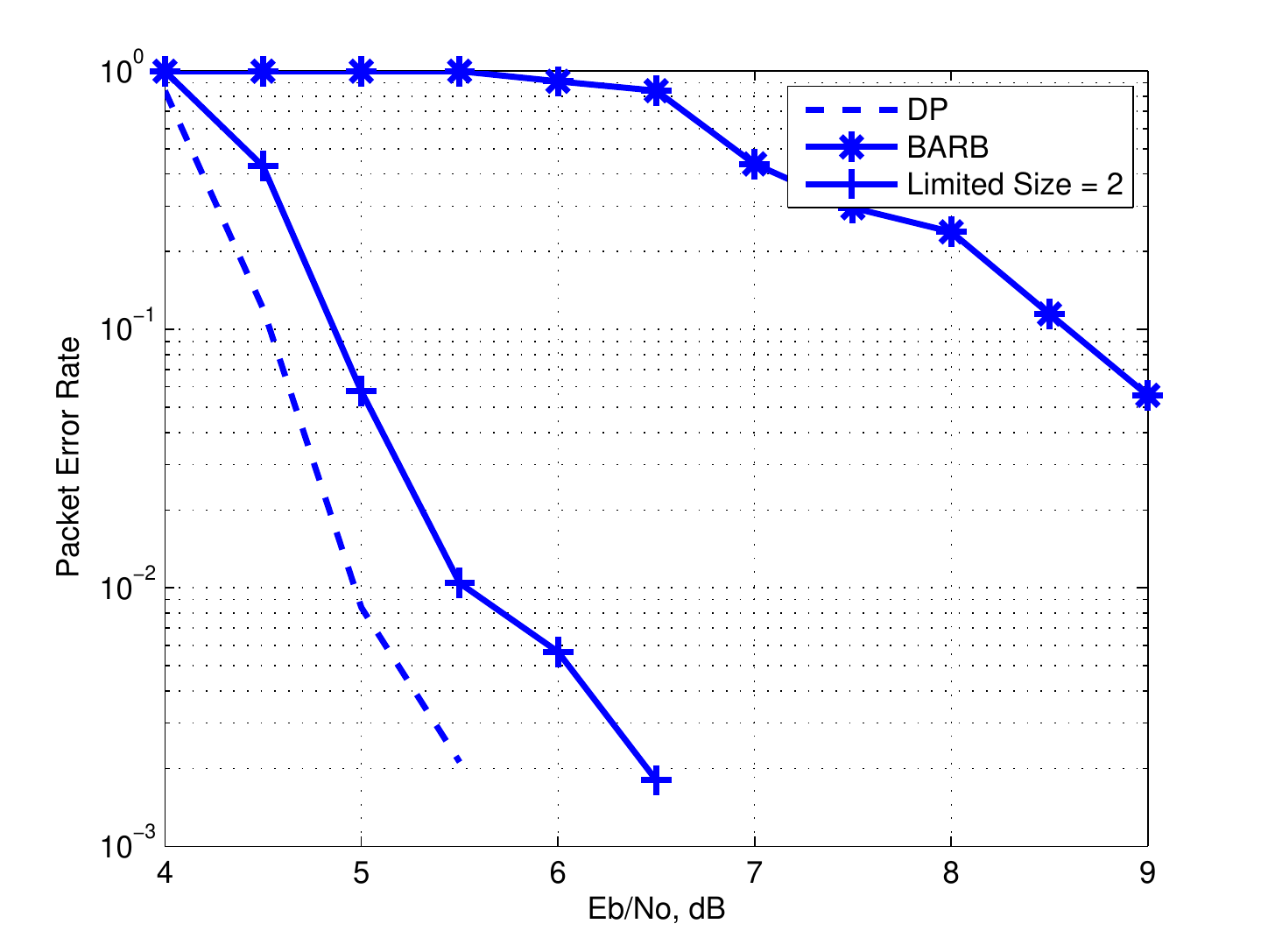}\\
  \caption{Packet Error Rate error rate - QPSK , $\sigma_{\Delta}=0.1$, Pilot Frequency $= 0.05$}\label{fig:perqpsk}
\end{figure}

\begin{figure}
  \centering
  \includegraphics[width=8.5cm]{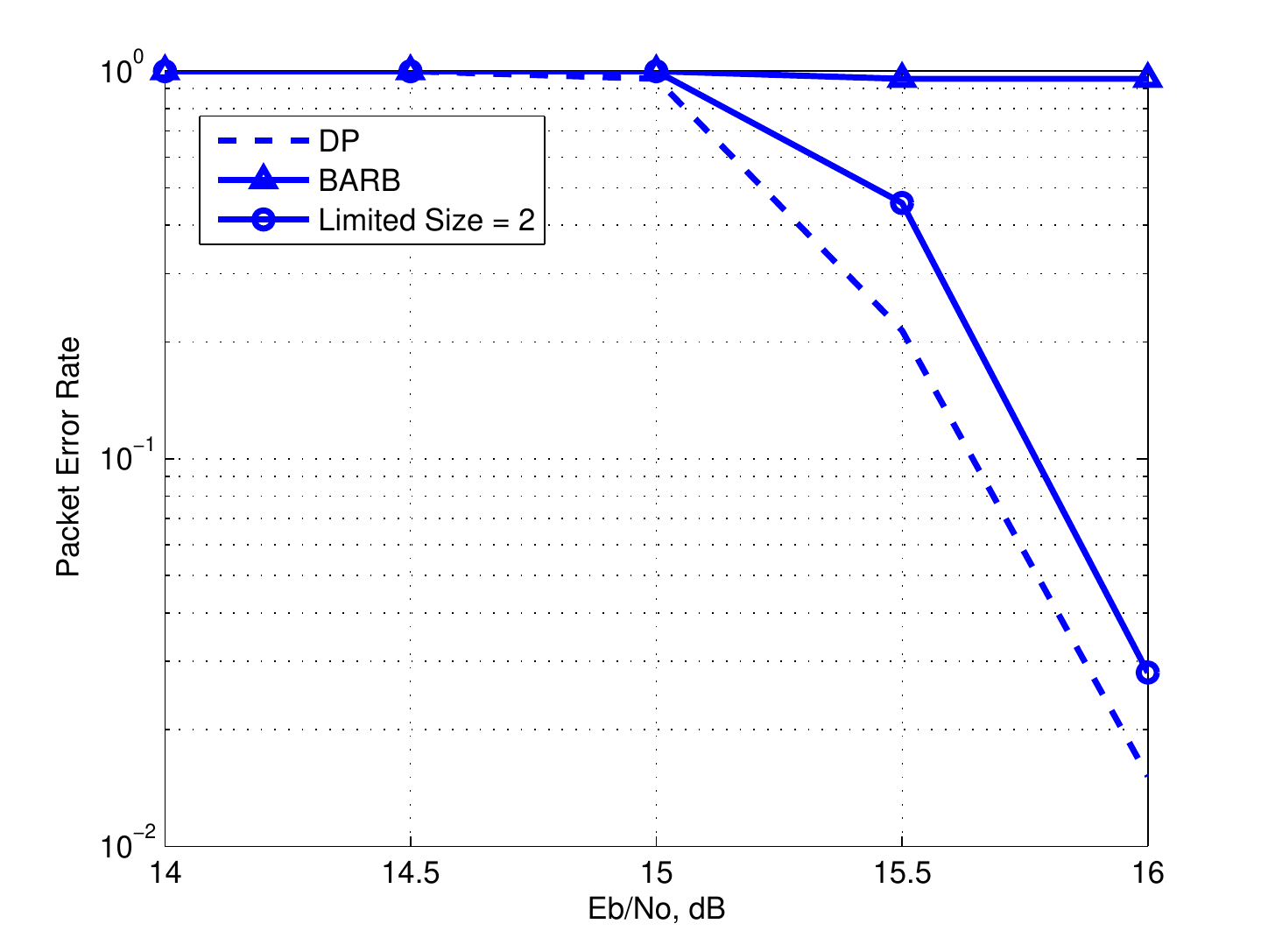}\\
  \caption{Packet error rate - 32PSK, $\sigma_{\Delta}=0.01$, Pilot Frequency $= 0.025$}\label{fig:per32psk}
\end{figure}

We can see that the mixture of order 2 is close to the
performance of the optimal algorithm, even when very few pilots are
present and the code rate and constellation order are high. One should also observe that for the 32PSK
scenario, the BARB algorithm demonstrates a high error floor. This is because of the large phase noise
variance and large spacing
between pilots which causes the SPA messages to become uniform
and thus do not provide information for the LDPC decoder. The high
code rate amplifies this problem. However, the limited algorithm with only one Tikhonov component performs
almost as well as the DP algorithm. This is due to the cycle slip recovery procedure we have presented
earlier which enables the limited algorithm to regain tracking even after missing the correct trajectory.\\
In Fig. \ref{fig:avg_comps8_3lobes_eps4} we present the average number
of mixture components, for different SNR and
LDPC iterations for $\epsilon = 4$. It can be seen that for the first iteration, many components are
needed since there is a high level of phase
ambiguity. As the iterations progress the LDPC decoder sends better soft information for
the code symbols, resolving these ambiguities. Therefore, the average
number of mixture components becomes closer to $1$.

\begin{figure}
  \centering
  \includegraphics[width=8.5cm]{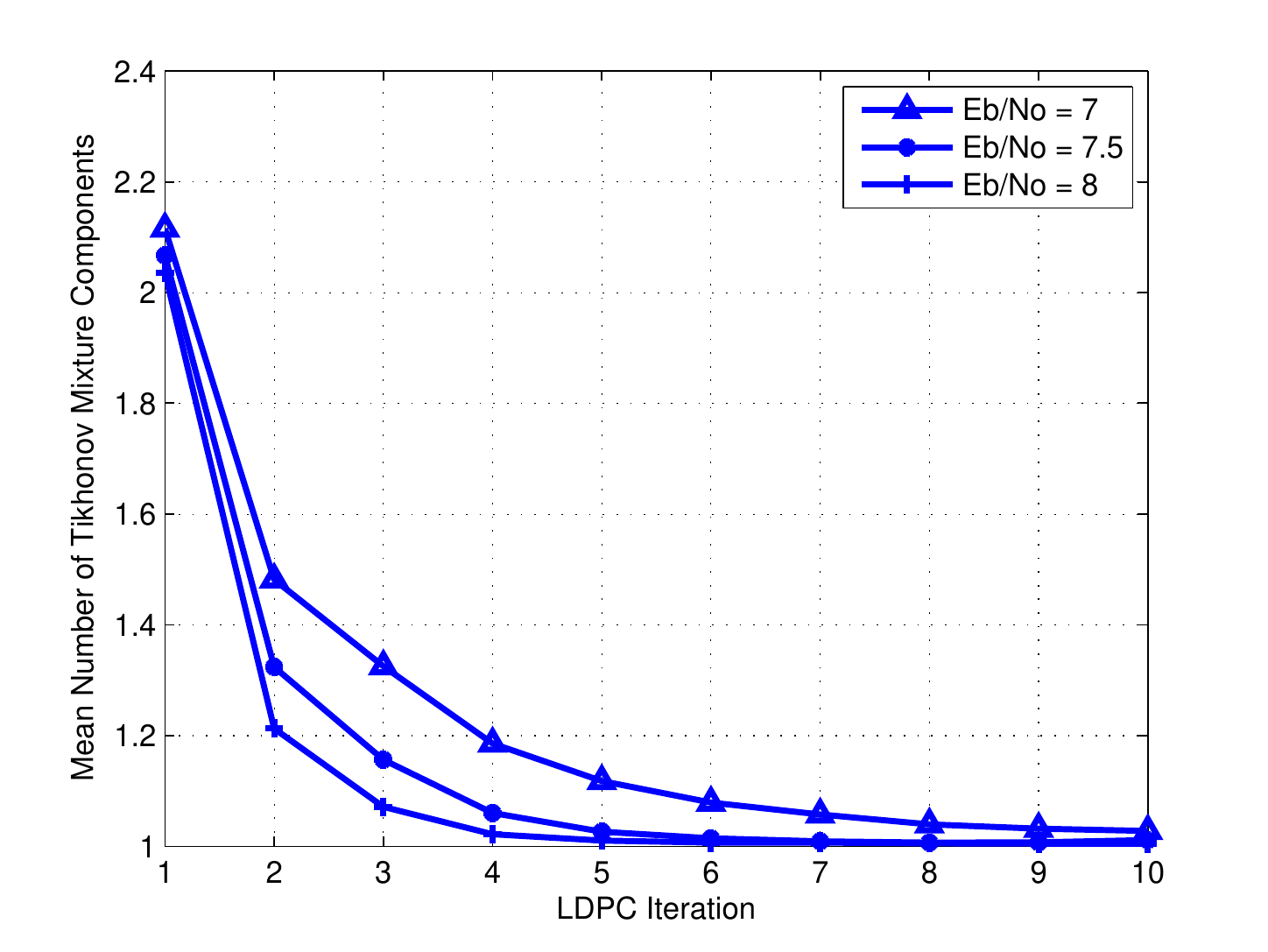}\\
  \caption{8PSK Mean Number of Tikhonov Mixture Components - Full Algorihtm, Maximum 3 lobes}\label{fig:avg_comps8_3lobes_eps4}
\end{figure}

\begin{figure}
  \centering
  \includegraphics[width=8.5cm]{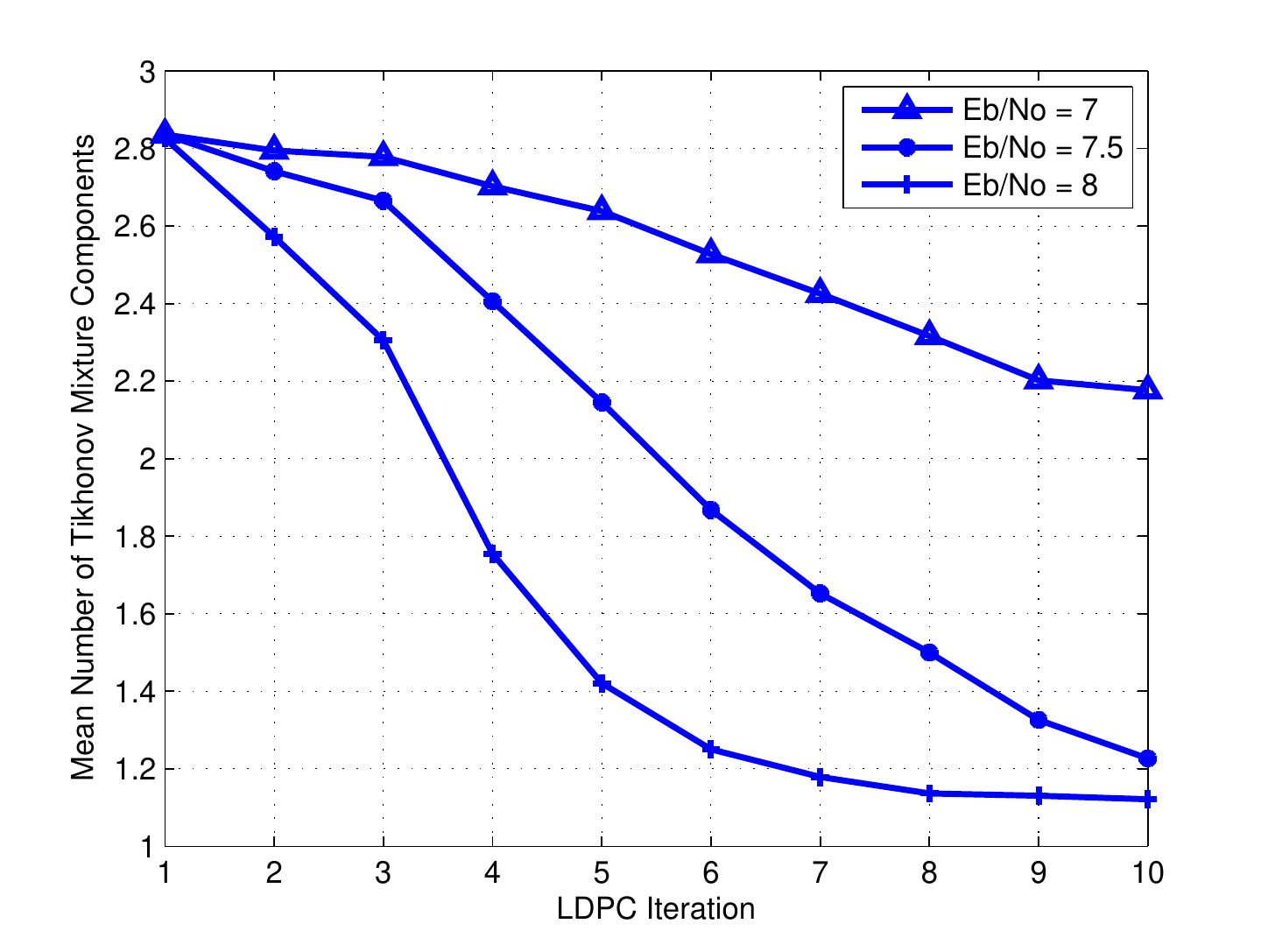}\\
  \caption{8PSK Mean Number of Tikhonov Mixture Components - Reduced Complexity Algorithm, Maximum 3lobes}\label{fig:avg_comps8_3lobes_eps1}
\end{figure}

In Fig. \ref{fig:avg_comps8_3lobes_eps1} we present the average number
of mixture components for the reduced complexity algorithm, for different SNR and
LDPC iterations for $\epsilon = 1$. We chose $\epsilon$ to be lower since we do not use the CMVM operator
as described earlier. As shown in this figure, the mean number of components is larger than for $\epsilon= 4$ but the overall complexity is still manageable.
In Table (\ref{tab:complexity_8psk_reduced}), the computational complexity of the reduced complexity
algorithm is compared to the DP and BARB algorithms. We use the mean mixture in Fig. \ref{fig:avg_comps8_3lobes_eps1} as $\gamma$. We can see that the algorithms proposed in this contribution,
have extremely less computational complexity than DP, while having comparable PER levels to it.

\begin{table*}[!ht] 
\caption{Simulation Results - Computational load per code symbol for 8PSK constellation at $\frac{E_{b}}{N_{0}}=8dB$}  
\centering  
\begin{tabular}{p{2cm}|p{3cm}|p{3cm}|p{3cm}}                       
 Algorithm &DP & BARB & Reduced Complexity, Order 3\\
 \hline
 Iteration & Constant for all iterations &Constant&1 2 3 4
\\
\hline\hline             
MULS & $68360$& $61$& $312$ $292$ $273$ $238$\\

LUT & $128$& $24$& $147$ $134$ $123$ $102$\\
\hline                          
\end{tabular}
\label{tab:complexity_8psk_reduced}
\end{table*}

It should be noted, that the PER performance of the Unlimited algorithm, for small enough $\epsilon$, is
as good as the PER performance of the DP algorithm because the mixture algorithm tracks all the
significant trajectories with no limit on the mixture order. The choice of the threshold $\epsilon$ in the
algorithm is according to the level of distortion allowed for the reduced mixture with respect to the
original mixture. If $\epsilon$ is very close to zero, then there will not be any components close enough
and the mixture will not be reduced. Therefore, there is a tradeoff between complexity and accuracy in the
selection of this parameter. This tradeoff is illustrated in Fig. \ref{fig:avg_comps8_eps_1_4}, where we
have plotted the mean mixture order for the unlimited algorithm using $\epsilon=1$ and $\epsilon=4$. It
should be noted that for these values and chosen SNRs, the unlimited algorithm has the same PER levels for
both $\epsilon$. However, choosing $\epsilon=15$ with the same algorithm will increase the PER. Therefore,
choosing the threshold too low might increase the mixture order with no actual need.
\begin{figure}
  \centering
  \includegraphics[width=8.5cm]{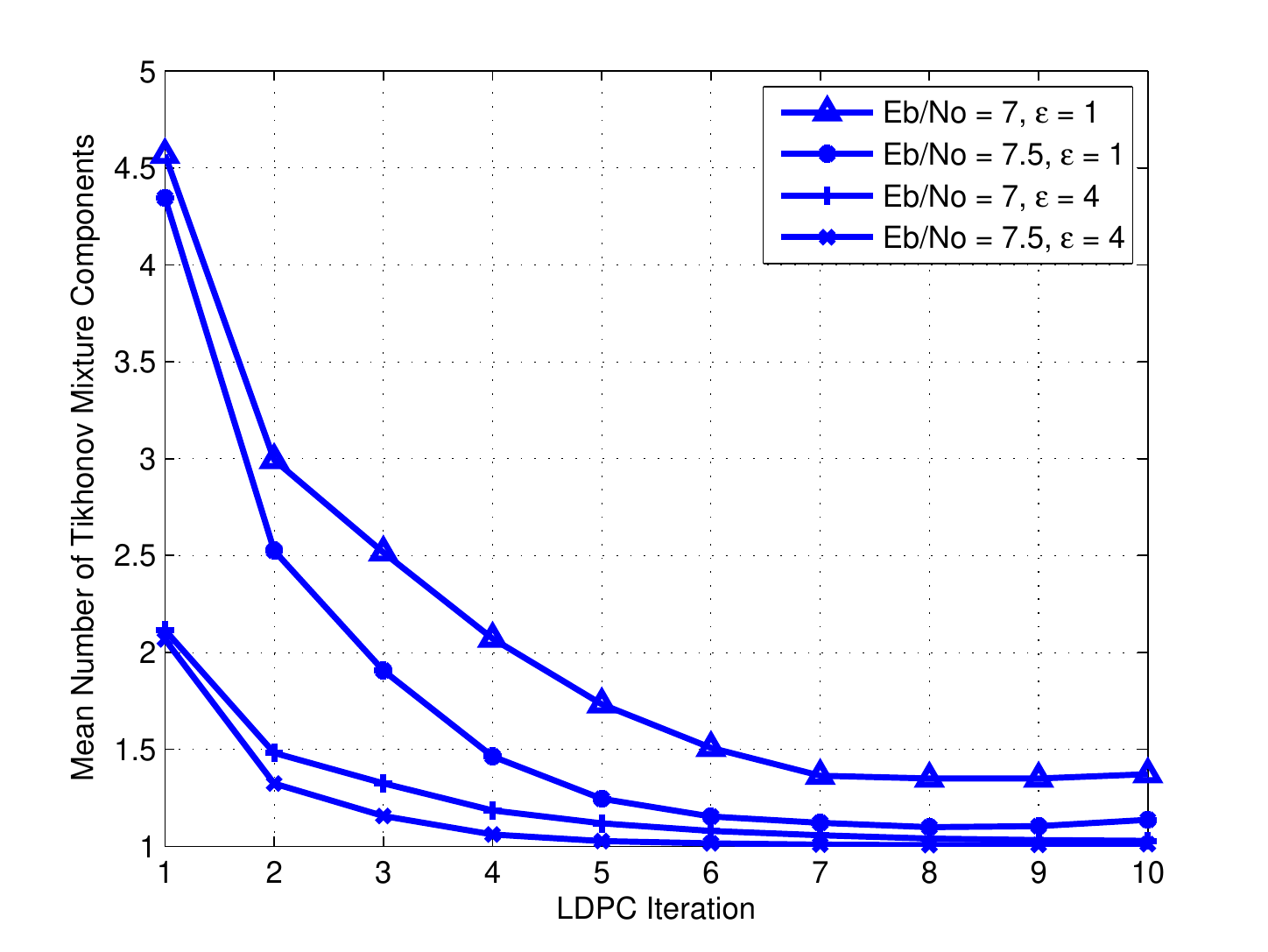}\\
  \caption{8PSK Mean Number of Tikhonov Mixture Components - Unlimited Algorithm}\label{fig:avg_comps8_eps_1_4}
\end{figure}

\section{Discussion}
In this paper we have presented a new approach for joint decoding and estimation of LDPC coded
communications in phase noise channels. The proposed algorithms are based on the approximation of SPA
messages using Tikhonov mixture
canonical models. We have presented an innovative approach for  mixture dimension reduction which keeps
accuracy levels high and is low complexity. The decoding scheme proposed in this contribution is shown via
simulations to significantly reduce the computational complexity of the best known decoding algorithms,
while keeping PER levels very close to the optimal algorithm (DP).
Moreover, we have presented a new insight to the underlying dynamics of phase noise estimation using
Bayesian methods. We have shown that the estimation algorithm can be viewed as trajectory tracking, thus
enabling the development of the mixture reduction and clustering algorithms which can be viewed as PLLs.

\appendices
\section{Proof of the CMVM Theorem}
\label{sec:cmvm_thr}
Let $f(\theta)$ be any circular distribution defined on $[0,2\pi)$ and
$g(\theta)$ a Tikhonov distribution.
\begin{equation}\label{22}
   g(\theta) = \frac{e^{Re[\kappa e^{-j(\theta-\mu)}]}}{2\pi I_{0}(\kappa)}
\end{equation}

We wish to find,
\begin{equation}\label{KL1}
    [\mu^{*},\kappa^{*}] = arg\min_{\mu,\kappa} D_{KL}(f||g)
\end{equation}

According to the definition of the KL divergence,
\begin{equation}\label{23}
   D_{KL}(f||g) = -h(f)-\int^{2\pi}_{0}f(\theta)\log g(\theta)d\theta
\end{equation}

where the differential entropy of the circular distribution
$f(\theta)$, $h(f)$ does not affect the optimization,
\begin{equation}\label{KL2}
    [\mu^{*},\kappa^{*}] = arg\max_{\mu,\kappa}\int^{2\pi}_{0}f(\theta)\log g(\theta)d\theta
\end{equation}
After the insertion of the Tikhonov form into (\ref{KL2}), we get
\begin{equation}\label{KL3}
    [\mu^{*},\kappa^{*}] = arg\max_{\mu,\kappa}\int^{2\pi}_{0}f(\theta)Re[\kappa e^{-j(\theta-\mu)}]d\theta\ -\log{2\pi I_{0}(\kappa)}
\end{equation}

Rewriting (\ref{KL3}) as an expectation and maximizing over $\mu$ only,
\begin{equation}\label{24}
   \mu^{*} = arg\max_{\mu} \kappa \mathbb{E}(Re[ e^{-j(\theta-\mu)}])
\end{equation}

Using the linearity of the expectation and real operators,
\begin{equation}\label{KL4}
    \mu^{*} = arg\max_{\mu} \kappa Re[\mathbb{E}(e^{j(\theta-\mu)})]
\end{equation}
We can view (\ref{KL4}) as an inner product operation and therefore,
the maximal value of $\mu$ is obtained, according to the
Cauchy-Schwartz inequality, for
\begin{equation}\label{25}
   \mu^{*} = \angle{\mathbb{E}(e^{j(\theta)})}
\end{equation}

Now we move on to finding the optimal $\kappa$, using the fact that we
found the optimal $\mu$.
For $\mu^{*}$, the optimal $g(\theta)$ needs to satisfy
\begin{equation}\label{diff0}
    \frac{\partial D(f||g)}{\partial \kappa} = 0
\end{equation}
After applying the partial derivative to (\ref{KL3}), and using
\begin{equation}\label{26}
   \frac{dI_{0}(\kappa)}{d\kappa}\ = \frac{I_{1}(\kappa)}{I_{0}(\kappa)}
\end{equation}

We get,
\begin{equation}\label{27}
   \mathbb{E}(Re[e^{-j(\theta-\mu^{*})}]) =
\frac{I_{1}(\kappa^{*})}{I_{0}(\kappa^{*})}
\end{equation}

Recalling (\ref{circ_mu}) and (\ref{circ_var}), we get that the
optimal Tikhonov distribution $g(\theta)$ is given by matching its
circular mean and variance to the circular mean and circular variance
of the distribution $f(\theta)$.
\qed

\section{Using The \textsf{CMVM} Operator to Cluster Tikhonov Mixture Components}
\label{sec:cmvm_exmp}
In algorithms 1 \& 2, at each clustering iteration, a set $J$ of mixture components indices of the input
Tikhonov mixture (\ref{in}) is selected. The corresponding mixture components are clustered using the CMVM
operator. In this appendix we will explicitly compute the application of the CMVM operator and introduce
several approximations to speed up the computational complexity.
For simplicity, assume that the mixture components in the set $J$ are,
\begin{equation}\label{mix_tikh_11}
    f^{J}(\theta_{k})=\sum_{l \in J}^{|J|}\alpha_{l}\frac{e^{Re[Z_{l}e^{-j\theta_{k}}]}}{2\pi I_{0}(|Z_{l}|)}
\end{equation}

Using Theorem (\ref{mix_tikh_thr}) and skipping the algebraic details,
the \textsf{CMVM} operator for (\ref{mix_tikh_11}), is:
\begin{equation}\label{cmvm_res}
    \textsf{CMVM}(f^{J}(\theta_{k})) =\frac{e^{Re[Z^{f}_{k}e^{-j\theta_{k}}]}}{2\pi I_{0}(|Z^{f}_{k}|)}
\end{equation}

where
\begin{equation}\label{Zf}
    Z_{k}^{f} = \hat{k}e^{j\hat{\mu}}
\end{equation}

and
\begin{equation}\label{tikh_modify_eqs_mu}
    \hat{\mu} = \arg{{\sum_{l \in J}^{|J|}\alpha_{l}\frac{I_{1}(|Z_{l}|)}{I_{0}(|Z_{l}|)}e^{j\arg(Z_{l})}}}
\end{equation}
\begin{equation}\label{tikh_modify_eqs_k}
    \frac{1}{2 \hat{k}} = 1 - \sum_{l \in J}^{|J|}\alpha_{l}\frac{I_{1}(|Z_{l}|)}{I_{0}(|Z_{l}|)}Re[e^{j(\hat{\mu}-arg(Z_{l}))}]
\end{equation}

Since implementing a modified bessel function is computationally prohibitive, we present the following

approximation,
\begin{equation}\label{assumps3}
    \log(I_{0}(k)) \approx k-\frac{1}{2}\log(k)-\frac{1}{2}\log(2\pi)
\end{equation}
which holds for $k>2$, i.e. reasonably narrow distributions.

Using the following relation,
\begin{equation}\label{assumps1}
    I_{1}(x) = \frac{dI_{0}(x)}{dx}
\end{equation}

We find that,
\begin{equation}\label{fracI1I0}
    \frac{I_{1}(k)}{I_{0}(k)} = \frac{d}{dk}(\log(I_{0}(k)))
\end{equation}
Therefore
\begin{equation}\label{eq5}
    \frac{I_{1}(k)}{I_{0}(k)} \approx 1-\frac{1}{2k}
\end{equation}
Thus, the approximated versions of (\ref{tikh_modify_eqs_k}) and (\ref{tikh_modify_eqs_mu}) are
\begin{equation}\label{tikh_modify_eqs_mu_simp}
    \hat{\mu} = \arg[{{\sum_{l \in J}^{|J|}\alpha_{l}(1-\frac{1}{2|Z_{l}|})e^{j\arg(Z_{l})}}}]
\end{equation}
\begin{equation}\label{tikh_modify_eqs_k_simp}
    \frac{1}{2 \hat{k}} = 1 - \sum_{l \in J}^{|J|}\alpha_{l}(1-\frac{1}{2|Z_{l}|})\cos(\hat{\mu}-\arg(Z_{l}))
\end{equation}
We also use the approximation for the modified bessel function in the
computation of $\alpha_{l}$.
For a small enough $\epsilon$, $\cos(\hat{\mu}-\arg(Z_{l})) \approx 1$, thus one can further reduce the
complexity of (\ref{tikh_modify_eqs_k_simp})

\begin{equation}\label{tikh_modify_eqs_k_simp1}
    \frac{1}{\hat{k}} = \sum_{l \in J}^{|J|}\alpha_{l}\frac{1}{|Z_{l}|}
\end{equation}
which coincides with the computation of a variance of a Gaussian mixture.

\section{Computation of the KL divergence between two Tikhonov Distributions}
\label{sec:kl_comp}
In this section we will provide the computation of the KL divergence between two Tikhonov distributions,
which is a major part of both mixture reduction algorithms. We will also provide approximations used to
better the computational complexity of this computation.
Suppose two Tikhonov distributions $g_{1}(\theta)$ and $g_{2}(\theta)$, where

\begin{equation}\label{g1}
g_{1}(\theta) =  \frac{e^{Re[z_{1}e^{-j\theta}]}}{2\pi I_{0}(|z_{1}|)}
\end{equation}

\begin{equation}\label{g2}
g_{2}(\theta) =  \frac{e^{Re[z_{2}e^{-j\theta}]}}{2\pi I_{0}(|z_{2}|)}
\end{equation}

We wish to compute the following KL divergence,
\begin{equation}\label{28}
   D_{KL}(g_{1}(\theta) || g_{2}(\theta))
\end{equation}

which is,
\begin{equation}\label{29}
  D_{KL} = \int^{2\pi}_{0}g_{1}(\theta) \log(\frac{e^{Re[z_{1}e^{-j\theta}]}I_{0}(|z_{2}|)}{e^{Re[z_{2}e^{-j\theta}]}I_{0}(|z_{1}|)})d\theta
\end{equation}

Thus,
\begin{equation}\label{30}
  D_{KL} = \log(\frac{I_{0}(|z_{2}|)}{I_{0}(|z_{1}|)}) + \int^{2\pi}_{0}g_{1}(\theta) Re[z_{1} - z_{2}e^{-j\theta}]d\theta
\end{equation}

After some algebraic manipulations, we get
\begin{multline}\label{31}
  D_{KL} = \log(\frac{I_{0}(|z_{2}|)}{I_{0}(|z_{1}|)}) + \\ \frac{I_{1}(|z_{1}|)}{I_{0}(|z_{1}|)}(|z_{1}|-|z_{2}|cos(\angle z_{1} - \angle z_{2}))
\end{multline}

Using (\ref{eq5}) and (\ref{assumps3}) we get
\begin{multline}\label{kl_comp_full}
  D_{KL} \approx |z_{2}|(1-cos(\angle z_{1} - \angle z_{2}))- \\ \frac{1}{2}\log(\frac{|z_{2}|}{|z_{1}|})+\frac{|z_{2}|}{2|z_{1}|}cos(\angle z_{1} - \angle z_{2})
\end{multline}

\section{Proof of Mixture Reduction as Multiple PLLs}
\label{sec:pll}
In this section we will prove the claim presented in section \ref{sec:phase_tracking}, that under certain
channel conditions, the mixture reduction algorithms can be viewed as multiple PLLs tracking the different
phase trajectories. For reasons of simplicity, will only show the case where the mixture reduction
algorithm converges to a single PLL (the generalization for more than one PLL is trivial, as long as there
are no splits).
As described earlier, we model the forward messages as Tikhonov mixtures. Suppose the $m^{th}$ component
is,

\begin{equation}\label{pf_k}
    p^{m}_{f}(\theta_{k-1}) = \frac{e^{Re[z^{k-1,f}_{m} e^{-j\theta_{k-1}}]}}{2\pi
I_{0}(|z^{k-1,f}_{m}|)}
\end{equation}

then using (\ref{pf}), we get a Tikhonov mixture $f(\theta_{k})$,
\begin{equation}\label{pf_k_mul_pd}
    f(\theta_{k}) = \sum_{i=1}^{M}\alpha_{i}f_{i}(\theta_{k})
\end{equation}

where,
\begin{equation}\label{z_k_1}
f_{i}(\theta_{k}) = \frac{e^{Re[\tilde{z}^{k-1,f}_{m,i}e^{-j\theta_{k}}]}}{2\pi
I_{0}(|\widetilde{z}^{k-1,f}_{m,i}|)}
\end{equation}

\begin{equation}\label{tmp_z}
\tilde{z}^{k-1,f}_{m,i} = \frac{(z^{k-1,f}_{m} +
\frac{r_{k-1}x_{i}^{*}}{\sigma^2})}{1+\sigma^{2}_{\Delta}|(z^{k-1,f}_{m} +
\frac{r_{k-1}x_{i}^{*}}{\sigma^2})|}
\end{equation}

and $x_{i}$ is the $i^{th}$ constellation symbol.
We insert (\ref{pf_k_mul_pd}) into the mixture
reduction algorithms. Assuming slowly varying phase noise and high SNR, such that the mixture
reduction will cluster all the mixture components, with non negligible
probability, to one Tikhonov distribution. Then, the \emph{circular} mean, $\hat{\theta}_{k}$, of the
clustered
Tikhonov distribution is computed according to,

\begin{equation}\label{mu_k_1}
    \hat{\theta}_{k} = \angle \mathbb{E} (e^{j\theta_{k}})
\end{equation}

where the expectation is over the distribution $f(\theta_{k})$.
We note that for every complex valued scalar $z$, the following holds
\begin{equation}\label{im_log}
    \angle{z} = \Im(\log{z})
\end{equation}

where $\Im$ denotes the imaginary part of a complex scalar.
If we apply (\ref{im_log}) to (\ref{mu_k_1}) we get,
\begin{equation}\label{im_log1}
    \hat{\theta}_{k} =  \Im{\left(\log{\sum_{i=1}^{M}\alpha_{i}\frac{\widetilde{z}^{k-1,f}_{m,i}}{|\widetilde{z}^{k-1,f}_{m,i}|}}\right)}
\end{equation}

which can be rewritten as,
\begin{equation}\label{im_log2}
    \hat{\theta}_{k} =  \Im{\left(\log{\sum_{i=1}^{M}\alpha_{i}\frac{z^{k-1,f}_{m} +
\frac{r_{k-1}x_{i}^{*}}{\sigma^2}}{|{z^{k-1,f}_{m} +
\frac{r_{k-1}x_{i}^{*}}{\sigma^2}}|}}\right)}
\end{equation}

we denote,
\begin{equation}\label{33}
  G_{k-1} = |z^{k-1,f}_{m} +
\frac{r_{k-1}x_{i}^{*}}{\sigma^2}|
\end{equation}

and assume that $G_{k-1}$, the conditional causal MSE of the
phase estimation under mixture component $f_{i}(\theta_{k})$, is constant for all significant components.
Then,
\begin{equation}\label{34}
  \hat{\theta}_{k} \approx  \hat{\theta}_{k-1} +
\Im\left(\log\left({\sum_{i=1}^{M}\alpha_{i}\left(1+\frac{r_{k-1}x_{i}^{*}}{G_{k-1}z^{k-1,f}_{m}\sigma^2}\right)}\right)\right)
\end{equation}

where,
\begin{equation}\label{theta}
  \hat{\theta}_{k-1} = \angle z^{k-1,f}_{m}
\end{equation}

\begin{equation}\label{35}
  \hat{\theta}_{k} \approx  \hat{\theta}_{k-1} +
\Im\left(\log\left({1+\frac{r_{k-1}}{G_{k-1}z^{k-1,f}_{m}\sigma^2}\left(\sum_{i=1}^{M}\alpha_{i}x_{i}^{*}\right)}\right)\right)
\end{equation}

We will define $c_{soft}$ as the soft decision symbol using the significant components,
\begin{equation}\label{36}
  c_{soft} = \sum_{i=1}^{M}\alpha_{i}x_{i}
\end{equation}

Since we assume high SNR and small phase noise variance, then the
tracking conditional MSE will be low, i.e $|z^{k,f}_{1}|$ will be high.
Using the fact that for small angles $\phi$,

\begin{equation}\label{37}
 \angle(1+\phi) \approx \Im(\phi)
\end{equation}
Therefore,
\begin{equation}\label{log_simp}
    \hat{\theta}_{k} \approx  \hat{\theta}_{k-1} +
\Im(\frac{r_{k-1}c^{*}_{soft}}{G_{k-1}z^{k-1,f}_{m}\sigma^2})
\end{equation}

Which, again for small angles $x$, $sin(x) \approx x$,
\begin{equation}\label{pll_equiv}
    \hat{\theta}_{k} \approx  \hat{\theta}_{k-1} +
\frac{|r_{k-1}||c^{*}_{soft}|}{G_{k-1}|z^{k-1,f}_{m}|\sigma^2}(\angle{r_{k-1}}+\angle{c^{*}_{soft}}-\hat{\theta}_{k-1})
\end{equation}













\bibliographystyle{plain}
\bibliography{strings}

\begin{thebibliography}{10}

\bibitem{colavolpe2006}
Giulio Colavolpe.
\newblock On ldpc codes over channels with memory.
\newblock {\em IEEE Transactions on Wireless Communications}, 5:1757 --1766,
  July 2006.

\bibitem{barb2005}
Giulio Colavolpe, Alan Barbieri, and Giuseppe Caire.
\newblock Algorithms for iterative decoding in the presence of strong phase
  noise.
\newblock {\em IEEE Journal on Selected Areas in Communications}, 23:1748
  --1757, September 2005.

\bibitem{sv2011}
David~F. Crouse, Peter Willett, Krishna Pattipati, and Lennart Svensson.
\newblock A look at gaussian mixture reduction algorithms.
\newblock In {\em Proceedings of the 14th International Conference on
  Information Fusion (FUSION)}, 2011.

\bibitem{Minh2003}
Minh~N. Do.
\newblock Fast approximation of kullback-leibler distance for dependence trees
  and hidden markov models.
\newblock {\em IEEE Signal Processing Letters}, 10:115 -- 118, April 2003.

\bibitem{goldberger2004hierarchical}
Jacob Goldberger and Sam Roweis.
\newblock Hierarchical clustering of a mixture model.
\newblock 2004.

\bibitem{KL1951}
Solomon Kullback and Richard~A. Leibler.
\newblock On information and sufficiency.
\newblock {\em The Annals of Mathematical Statistics}, 22:79--86, March 1951.

\bibitem{mardia2000}
Kanti~V. Mardia and Peter~E. Jupp.
\newblock {\em Directional Statistics}.
\newblock John Wiley and Sons Ltd., 2000.

\bibitem{runnalls2007}
Andrew~R. Runnalls.
\newblock Kullback-leibler approach to gaussian mixture reduction.
\newblock {\em IEEE Transactions on Aerospace and Electronic Systems}, 43:989
  --999, JULY 2007.

\bibitem{shachar2012}
Shachar Shayovitz and Dan Raphaeli.
\newblock Efficient iterative decoding of ldpc in the presence of strong phase
  noise.
\newblock In {\em Proceedings of The 7th International Symposium on Turbo Codes
  \& Iterative Information Processing}, 2012.

\bibitem{shachar_old2012}
Shachar Shayovitz and Dan Raphaeli.
\newblock Improved message passing algorithm for phase noise channels using
  optimal approximation of tikhonov mixtures.
\newblock In {\em Proceedings of the 5th International Symposium on
  Communications, Control and Signal Processing, ISCCSP 2012, Rome, Italy, 2-4
  May 2012}, 2012.

\bibitem{shachar_multi2012}
Shachar Shayovitz and Dan Raphaeli.
\newblock Multiple hypotheses iterative decoding of ldpc in the presence of
  strong phase noise.
\newblock In {\em Proceedings of The 2012 IEEE 27th Convention of Electrical
  and Electronics Engineers in Israel}, 2012.

\bibitem{worthen2001}
Andrew~P. Worthen and Wayne~E. Stark.
\newblock Unified design of iterative receivers using factor graphs.
\newblock {\em IEEE Transactions on Information Theory}, 47:843 --849, February
  2001.

\end{thebibliography}

\end{document}